\documentclass[12pt,a4paper]{article}

\usepackage[text={6in,8.5in}]{geometry}
\usepackage{graphicx}
\usepackage{amsmath}
\usepackage{amsfonts}
\usepackage{amsthm}
\usepackage{amssymb}
\usepackage{booktabs}
\usepackage{float}
\usepackage[framemethod=tikz]{mdframed}
\usepackage{subcaption}
\usepackage{microtype}
\usepackage[backend=biber,style=numeric]{biblatex}
\usepackage{tikz}
\tikzstyle{fsquare} = [rectangle,inner sep=2pt,outer sep=3pt,fill]
\tikzstyle{fcircle} = [circle,inner sep=1.5pt,outer sep=3pt,fill]
\tikzstyle{hcircle} = [circle,inner sep=1.5pt,outer sep=3pt,draw]

\usepackage[T1]{fontenc}

\usepackage[ruled,vlined]{algorithm2e}

\usepackage{aligned-overset}
\usepackage{parskip}

\usepackage{pgfplots}

\newtheorem{thm}{Theorem}[section]
\newtheorem{lem}[thm]{Lemma}

\newtheorem{cor}[thm]{Corollary}

\theoremstyle{remark}
\newtheorem{remark}{Remark}

\theoremstyle{definition}
\newtheorem{definition}[thm]{Definition}
\newtheorem{case}{Case}
\newtheorem{subcase}{Case}

\newenvironment{customproblem}[3]{
        \begin{mdframed}[frametitle={\textsc{#1}}]
                \begin{tabular}{p{0.12\columnwidth} p{.8\columnwidth}}
                        \textbf{Instance:} & #2 \\
                        \textbf{Task:} & #3
                \end{tabular}
        \end{mdframed}
}

\numberwithin{equation}{section}

\newcommand{\bN}{\mathbb{N}}
\newcommand{\bR}{\mathbb{R}}
\newcommand{\cT}{\mathcal{T}}
\newcommand{\cA}{\mathcal{A}}

\DeclareMathOperator{\cost}{cost}

\DeclareMathOperator{\OPT}{OPT}

\makeatletter
\newcommand{\nosemic}{\renewcommand{\@endalgocfline}{\relax}}
\newcommand{\dosemic}{\renewcommand{\@endalgocfline}{\algocf@endline}}
\let\oldnl\nl
\newcommand{\nonl}{\renewcommand{\nl}{\let\nl\oldnl}}
\makeatother

\begin{document}

\title{Further Improvements on  Approximating  the \\Uniform Cost-Distance Steiner Tree Problem}

\author{Stephan Held and   Yannik Kyle Dustin Spitzley  ~\vspace{10pt} \\
  Research Institute for Discrete Mathematics and \\Hausdorff Center for Mathematics\\ University of Bonn ~\vspace{10pt}\\
\texttt{\href{mailto:held@dm.uni-bonn.de}{held@dm.uni-bonn.de} \href{mailto:spitzley@dm.uni-bonn.de}{spitzley@dm.uni-bonn.de}}}
\date{November 6, 2022}

\maketitle

\begin{abstract}
  In this paper, we consider the \textsc{Uniform Cost-Distance
    Steiner Tree Problem}  in  metric spaces,   a generalization of the well-known
  Steiner tree problem. Cost-distance Steiner trees
  minimize the sum of  the total length and the weighted path lengths
  from a dedicated root to the other terminals, which have a weight to penalize the path length.
  They are applied when the tree is intended for signal transmission, e.g.\
  in chip design or telecommunication networks, and
  the signal speed through the tree has to be considered besides the total length.

  Constant factor approximation algorithms for the uniform cost-distance Steiner tree problem have been
  known since the first mentioning of the problem by Meyerson,
  Munagala, and Plotkin \cite{meyerson2008cost}.  Recently, the
  approximation factor was improved from $2.87$ to $2.39$ by Khazraei
  and Held \cite{Khazraei21}.  We refine  their approach further and
  reduce the approximation factor down to $2.15$.
\end{abstract}

~\newpage

\section{Introduction}

Steiner trees can be found in numerous applications, in particular in
chip  design and telecommunications.  In these applications,
both the total tree length and the signal speed are important.  We
consider Steiner trees that do not only minimize the total cost, but also the weighted path lengths from a dedicated root $r$ to the
other terminals. These weights arise for instance  as Lagrangean multipliers
when optimizing global signal delay constraints on an integrated circuit \cite{held2017global}.
Formally, the  problem is defined as follows.

\begin{customproblem}
{Uniform Cost-Distance Steiner Tree Problem
}{
        A metric space $(M,c)$, a root $ r$, a finite set $ T $ of sinks, a map $p:T\dot{\cup}\{r\}\to M$,  sink delay weights $ w\colon T\to \bR _{\geq 0} $.
}{
        Find a Steiner tree $ A $ for $ T\cup \{ r\} $ with an extension $p: V(A) \setminus (T\cup\{r\}) \to M$ minimizing
        \begin{align}
                \label{eq98}
                \sum _{\{x,y\}\in E(A)} c(p(x),p(y)) + \sum _{t\in T} \left(w(t)\!\!\!\! \sum_{\{x,y\} \in E(A_{[r,t]})} \!\!\!\!c(p(x),p(y))\right),
        \end{align}

        where $ A_{[r,t]} $ is the unique $ r $-$ t $-path in $ A $.
}
\end{customproblem}

Given a Steiner tree $A$, we call
$$\sum _{\{x,y\}\in E(A)} c(p(x),p(y)) $$ its \textbf{connection cost}
and $$\sum _{t\in T} \left(w(t)\!\!\!\! \sum_{\{x,y\} \in E(A_{[r,t]})} \!\!\!\!c(p(x),p(y))\right)$$ its \textbf{delay cost}.

Usually the position $p$ of vertices  is clear from the
context. Then, we simply write $c(x,y)$ instead of $c(p(x),p(y))$ and
$c(e)$ instead of $c(x,y)$ for edges $\{x,y\}$.
After orienting its edges, we can consider $A$ as an $r$-arborescence, reflecting the
signal direction in the underlying  applications. We often use arborescences instead of trees to simplify the notation.

A simple lower bound for the objective function, against which we will compare, is given by
\begin{align}
  \label{eq:lower_bound}
  C_{SMT} + \sum _{t\in T} w(t)  c(r,t),
\end{align}
where $ C_{SMT}$ is the connection cost  of minimum-length Steiner tree for $T\cup\{r\}$.

The \textsc{Uniform Cost-Distance Steiner Tree Problem}  was first mentioned by \cite{meyerson2008cost}, who considered the
(general) cost-distance Steiner tree problem, where the connection cost
is measured using a different metric than the delay cost.  The general
cost-distance Steiner tree problem does not permit an approximation factor better
than $\Omega(\log \log |T|)$ unless $\text{NP} \subseteq
\text{DTIME}(|T|^{\mathcal{O}(\log\log\log |T|)})$ \cite{Chuzhoy08}, while a randomized
$\mathcal{O}(\log |T|$)-factor approximation algorithm was given by \cite{meyerson2008cost} and
 a derandomization in \cite{Chekuri01}.

Meyerson, Munagala and Plotkin \cite{meyerson2008cost} observed  that a constant factor approximation
algorithm can be obtained for the \textsc{Uniform Cost-Distance Steiner Tree Problem}  using the  shallow-light spanning tree algorithm from \cite{khuller1995balancing}.
The resulting factor is $3.57$. Using shallow-light Steiner trees \cite{held2013shallow} instead of spanning trees, the factor was improved
to $2.87$ independently by \cite{guo2014approximating} and \cite{rotter2017timing}.
So far the best factor  was given by Khazraei and Held \cite{Khazraei21}:

\begin{thm}[Khazraei and Held \cite{Khazraei21}]
        \label{theo99}
        The \textsc{Uniform Cost-Distance Steiner Tree Problem} can be solved in polynomial time within an approximation factor of $ (1+\beta ) $, where $ \beta $ is the approximation ratio for computing a minimum length Steiner tree in $(M,c)$.
\end{thm}

The best known Steiner approximation factor in general metric spaces is $ \beta = \ln (4) + \epsilon $ \cite{Byrka13,traub2022local}. Then,  Theorem~\ref{theo99} yields an approximation factor of $2.39$.
For the Euclidean or Manhattan plane, $\beta = (1+\epsilon)$ is possible for any $\epsilon >0$ \cite{Arora98}, leading to
an approximation factor of $2+\epsilon$ for the \textsc{Uniform Cost-Distance Steiner Tree Problem} in these spaces.

Similar to algorithms for shallow-light trees, the algorithm in \cite{Khazraei21} starts from an approximately minimum Steiner tree,
which is cut into a forest whose components are connected to the root $r$ individually.
While \cite{khuller1995balancing}  cut the tree  whenever the path length is too large,
\cite{Khazraei21} cut off a subtree   if its delay weight exceeds a certain threshold.
Each cut-off tree is later re-connected through a direct connection/path from the source
through one of its terminals that minimizes the resulting objective function.

Cost-distance Steiner trees are heavily  used in
VLSI routing and interconnect optimization \cite{held2017global,daboul2019global} to achieve a low
power consumption and, simultaneously, a fast chip speed.

The \textsc{Uniform Cost-Distance Steiner Tree Problem} is related to the \textit{single-sink buy-at-bulk problem}
where a set of demands needs to be routed from a set of sources to a single sink using a pipe network that has to be constructed from
a finite set of possible pipe types with different costs and capacities \cite{guha2009constant,talwar2002single,jothi2009improved}.
The best known approximation factor for this problem is  $40.82$ due to \cite{grandoni2010network}, who
also achieve a factor of $20.41$ for the splittable case.
If there is only one pipe type this problem is equivalent to the  \textsc{Uniform Cost-Distance Steiner Tree Problem}.
In fact, the threshold-based tree cutting used to proof Theorem~\ref{theo99} in \cite{Khazraei21}  is similar to the algorithm in  \cite{guha2009constant},
but the re-connect to the root/sink differs.

\subsection{Our contribution}

In this paper we improve the approximation algorithm in  \cite{Khazraei21} for the  \textsc{Uniform Cost-Distance Steiner Tree Problem}.

We show a class of instances, where the approximation factor $(1+\beta)$ in  \cite{Khazraei21} is tight.
In this class  the trees that are cut off from the initial short Steiner tree have an unfavorable structure.
We cut such trees further into two or three parts to improve the quality.
Furthermore, we also improve the structure of the tree that contains the root $r$ after cutting.

We obtain an approximation factor of
$$  \beta + \frac{2b \beta}{\sqrt{4b\beta + (\beta - 1)^2} + \beta - 1}, $$
where  $b = \frac{1609\sqrt{1609} - 42427}{34992}$.
For $\beta = \ln(4)+\epsilon$ \cite{Byrka13} this results in an  approximation factor of 2.15
and for $\beta = 1$ this would give the factor $1.8$, clearly improving upon the corresponding factors
$2.39$ and $2.0$ in  \cite{Khazraei21}.

We can achieve the lower approximation ratio even with a faster running time.
While \cite{Khazraei21} obtain a running time of $\mathcal{O}(\Lambda + |T|^2)$,
where $\Lambda$ is the time to compute an initial $\beta$-approximate minimum Steiner tree,
our running time is  $\mathcal{O}(\Lambda + |T|)$ (assuming that the metric $c$ can be evaluated in constant time).

The remainder of this paper is structured as follows.
In Section~\ref{sec:approximation-algorithm} we will briefly summarize the algorithm behind Theorem~\ref{theo99}
and the main results used in its proof.
Then, in Section~\ref{sec:worst_case_example} we show an example that the analysis
in Theorem~\ref{theo99} is tight. A further example in the Manhattan plane is given in Appendix~\ref{appendix:tight-example-in-Manhattan-plane}.

This motivates the  improvement that we achieve by splitting the trees in   the aforementioned forest into up to two trees  in Section~\ref{sec:cut-into-two}.
In Section~\ref{sec:improve-root-component}, we show how to reduce cost of the component containing the root
after the cutting.
A splitting into up to three trees is described in Section~\ref{sec:cut-into-three} and leads
to a  further  improvement of the approximation factor at the cost of a more complex analysis.
Our main result is presented in Section~\ref{sec:main-result}, where  we combine
all improvement steps to obtain a  better  approximation algorithm for the
\textsc{Uniform Cost-Distance Steiner Tree Problem}.
We finish with conclusions in Section~\ref{sec:conclusion}.

\section{The $(1+\beta)$-approximation algorithm}
\label{sec:approximation-algorithm}

The algorithm in \cite{Khazraei21} is  shown in Algorithm~\ref{alg:khazraei21}.
It assumes $|T| \ge 2$, as $|T| = 1$ is trivially solved by a shortest path computation.

\begin{algorithm}[H]
  \textbf{Step 1 (Initial arborescence):}

  First, it computes a  $ \beta $-approximate minimum cost Steiner $r$-arborescence $A_0$ for $ T\cup \{ r\} $
  with outdegree 0 at all sinks in $T$ and outdegree 2 at all Steiner vertices in $ V(A_0)\setminus T$.
This can be achieved using  a $\beta$-approximate Steiner tree algorithm, orienting its edges to form
an $r$-arborescence and possibly introducing  new Steiner vertices at existing positions or replacing paths by single edges to obtain the desired leaf set and degree constraints without increasing the total length.

~

\textbf{Step 2 (split into branching):}

Then, it cuts $A_0$ into a branching as follows.
$A_0$ is traversed  bottom-up.
Whenever an edge $(x,y)\in E(A_0)$ is traversed, the algorithm checks
if the total sink delay weight in the sub-arborescence $(A_0)_y$ rooted at $y$ exceeds a certain threshold $\mu >
0$. That is, whenever $w(V((A_0)_y)\cap T) > \mu$, the edge $(x,y)$ is removed creating a new arborescence $(A_0)_y$ in the branching.
Let $A_r = (A_0)_r$ denote the final root component. In $A_r$, we delete edges and vertices until every leaf is a terminal in
$T\cup \{r\}$.

~

Let $\cA$ denote the set of all arborescences that were cut off from $A_0$  this way,
and let $\cT := \{ V(A')\cap T: A' \in \cA\}$ be the corresponding set of terminal sets.

~

\textbf{Step 3  (Re-connect arborescences):}

In a re-connect phase, each $A'\in \cA$ is  re-connected to $r$.
To this end, we select a vertex $t \in T' := V(A')\cap T$ that minimizes
the cost for serving the sinks in $T'$ through the $r$-arborescence $A'+(r,t)$,
i.e.\ we select a vertex $t \in T'$ as a ``\textbf{port}'' for $T'$  that minimizes
$$  c(r,t) + c(E(A'))  + \sum _{t'\in T'} w(t')\cdot (c(r,t) +  c(E(A'_{[t,t']}))).$$

Let $t_1,\dots,t_{|\cA|} \in T$ be the set of selected port vertices.
The algorithm returns  the union of the final branching and
the port connections $A_0  + \{(r,t_i):  i\in \{1,\dots,|\cA|\}$.

\caption{$(1+\beta)$-approximation algorithm by \cite{Khazraei21}.}
\label{alg:khazraei21}
\end{algorithm}

It should be noted that the threshold-based tree cutting in Step 2
can be replaced by an optimum cutting (w.r.t.  the subsequent Step 3)
 using  a polynomial time dynamic program \cite{dynamic-program-splitting}.
However, no better approximation factors have been proven for such an optimum subdivision.

\subsection{Notation}

We will use the following notation throughout this paper.
Let $A$ be an arborescence. By   $ A_v $ we denote the sub-arborescence rooted at $v$.
Furthermore, $ T_{A} := V(A)\cap T $ is the set of terminals in $A$,
$ W_{A} := w(T_{A}) $ is the sum of delay weights in $A$,
$ C_{A} := c(E(A)) $ is  the \textbf{connection cost} of $A$ and
$ D_{A} := D_{T_A} := \sum_{t\in T_{A}} w(t)  c(r,t) $ the \textbf{minimum possible delay cost} for connecting the sinks in  $ T_{A} $.

\begin{definition}
Removing an edge $e\in E(A)$  splits $A$ into two arborescences. By
 \begin{align}
  \omega_{A}
  := \max _{(x,y)\in E(A)} W_{A_y} (W_{A}-W_{A_y})
  \label{eq:omega_A}
\end{align}
we denote the maximum product of weights in  these two arborescences that
can be obtained by removing an edge from $A$. We call $  \omega_{A}$ the \textbf{balance} of $A$.
\label{def:omega_A}
\end{definition}

Informally, the balance $ \omega _A $ indicates how regularly the delay weights are distributed in the arborescence $ A $.
If $ \omega _A $ is small, then for each edge $ e = (x,y) $ the sum of the delay weights in one of the two arborescences  in $ A-e $ is significantly larger than in the other.
On the other hand, if $ \omega _A $ is large, there is an edge $ e = (x,y) $ such that the delay weight in the to components $ W_{A_y} $ and $ W_{A} - W_{A_y} $ are similar, i.e. balanced.

\subsection{Essential steps in the proof of Theorem~\ref{theo99}}

We quickly recap the essential steps in the analysis of \cite{Khazraei21}.
The cost to connect an arborescence $A'\in\cA$ to the root $ r $ can be estimated as follows:

\begin{lem}[Khazraei and Held \cite{Khazraei21}, Lemma 1]
  \label{lem99}
  Let $A'  \in \cA$  with terminal corresponding set $ T' $. By the choice of the port $ t\in T'$, the $r$-arborescence  $(A' + \{r,t\})$  has a total cost at most
  \begin{align}
    \label{eq99}
    \left( 1+\frac{W_{A'}}{2} \right) C_{A'} + \left( 1 + \frac{1}{W_{A'}}\right) D_{T'} \le     \left( 1+\mu \right) C_{A'} + \left( 1 + \frac{1}{\mu}\right) D_{T'}.
  \end{align}

  If $ |T'|\ge 2 $,   $A' + \{r,t\}$  has a total cost at most
  \begin{align}
        \label{eqn:kh_bound_multi-terminals}
    \left( 1 + \frac{2\omega_{A'}}{W_{A'}}\right) C_{A'} + \left( 1 + \frac{1}{W_{A'}}\right) D_{T'} \le     \left( 1+\mu \right) C_{A'} + \left( 1 + \frac{1}{\mu}\right) D_{T'}.
  \end{align}
\end{lem}
\begin{proof}
Note that Lemma 1 in \cite{Khazraei21}  states only the common right hand side bound.
The left hand side bounds follow immediately from their proof, which  we will briefly sketch: The case of a single terminal $ T' = \{ t\} $ is quite simple, since $ C_A' = 0 $ and $ W_{A'} = w(t) > \mu $. For $ |T'|\geq 2 $, we choose a terminal $ t\in T' $ randomly with probability $ p_t := \frac{w(t)}{W_{A'}} $ as the ``port'' vertex (only for the analysis). Then we obtain:

\begin{itemize}
\item The  expected cost of $(r,t)$ is  $\mathbb{E}(c(r,t)) = \sum _{t\in T'} p_t c(r,t) = \frac{1}{W_{A'}} D_{T'} $.

\item The (deterministic) connection cost within $ A' $ is  $ C_{A'} $.

\item The expected  effective  delay cost of $(r,t)$ is $$ \mathbb{E}(W_{A'} \cdot c(r,t)) = W_{A'} \cdot\sum _{t\in T'} p_t c(r,t) = D_{T'}.$$

\item The expected delay weight served by an edge $ (x,y)\in E(A') $ is
$$                        \frac{W_{A_y}}{W_{A'}}\cdot (W_{A'} - W_{A_y}) + \frac{W_{A'} - W_{A_y}}{W_{A'}}\cdot W_{A_y} = \frac{2 W_{A_y} (W_{A'} - W_{A_y})}{W_{A'}} \leq \frac{2\omega _{A'}}{W_{A'}}\leq \frac{W_{A'}}{2},
  $$

  where $ A_y $ is the sub-arborescence of $ A'-(x,y) $ containing $ y $. The formula reflects the expected component of  $(A'-(x,y))$
  in which  the port vertex is located.
                Summation over all edges in $ A' $ yields the following expected delay cost contribution of $E( A')$: $$ \frac{2\omega _{A'}}{W_{A'}} C_{A'}\leq \frac{W_{A'}}{2} C_{A'}. $$
\end{itemize}

The addition of these four terms gives the expected total cost of connecting $ A' $ to the root $r$, and
provides the bounds in (\ref{eq99}) and (\ref{eqn:kh_bound_multi-terminals}). The deterministic best choice of the ``port'' vertex in Algorithm~\ref{alg:khazraei21} cannot be more expensive.
\end{proof}

A similar cost bound can  be shown easily for the arborescence $A_r$ containing the root $ r $
after Step 2. Summing up the resulting cost bounds, we get the following result:

\begin{thm}[Khazraei and Held \cite{Khazraei21}]
        \label{theo98}
        Given an instance $(M,c,r,T,p,w)$ of the \textsc{Uniform Cost-Distance Steiner Tree Problem}, a Steiner tree $ A $, whose objective value (\ref{eq98}) is at most
        \begin{align}
                \label{eq97}
                (1+\mu ) C + \left( 1 + \frac{1}{\mu}\right) D,
        \end{align}

        where $ C $ is the cost of a $ \beta $-approximate minimum Steiner tree and $ D := D_T $, can be computed in polynomial time.
\end{thm}

As $(\frac{1}{\beta}C + D)$ is a lower bound on the objective cost of any Steiner tree for $T\cup\{r\}$,  the approximation ratio $(1+\beta)$  in Theorem~\ref{theo99} follows by
choosing $\mu =
\frac{1}{\beta}$.

\section{Instance-specific choice of $\mu$ and worst-case examples}
\label{sec:worst_case_example}
Algorithm \ref{alg:khazraei21} achieves the approximation factor  in  (\ref{eq97})
or  $ 1+\beta $ for a choice of $\mu$ that depends only on $\beta$.
The terms $C$ and $D$ can easily be computed in Step 1 of Algorithm~\ref{alg:khazraei21}
using the initial arborescence $A_0$.

In a practical implementation, we would rather choose $\mu$ such that the bound in
(\ref{eq97}) is minimized directly.
If $D\not=0$ and $C\not=0$, the choice   $\mu = \sqrt{\frac{D}{C}} $ minimizes the upper bound
(\ref{eq97}) which becomes $C+D+2\sqrt{CD} \le (1+\beta)(C_{SMT} + D)$ , where $C_{SMT}$ is the
minimum length of a Steiner tree for $T \cup \{r\}$.

If $ C = 0 $, each $ r $-$ t $-path in $A_0$ after Step 1 ($ t \in T $)  has length $ 0 $.
Thus, $\mu=\infty$ (or suppressing Steps 2 and 3) yields an optimum solution.

If $ D = 0 $, then $ w(t) = 0 $ or $ c(r,t) = 0 $ for each $ t\in T $.
For each terminal $ t\in T $ with $ c (r,t) = 0 $, we add the edge $(r,t)$ (creating cycles).
A shortest-path $r$-arborescence in the resulting graph
is a  $ \beta $-approximate minimum Steiner tree with total  delay cost $ 0 $.
In particular, it is a  $ \beta $-approximate minimum cost-distance Steiner tree.

The next Theorem shows that even after these modifications the analysis cannot be improved.

\begin{thm}
        \label{lem95}
        For $ \beta = 1 $ and  $ \epsilon > 0 $, there is an instance of \emph{\textsc{Uniform Cost-Distance Steiner Tree Problem}}, for which Algorithm~\ref{alg:khazraei21}
        with the above modifications computes a solution with cost at least $ (2 - \epsilon ) \OPT $, where $ \OPT $ is the optimal objective function value (\ref{eq98}).
\end{thm}

\begin{proof}
  We define an instance class on a metric space that is induced by a graph $G$ with edge weights $c:E(G)\to \mathbb{R}_{\ge 0}$ as follows.
  Let $ k\in \mathbb{N}, k\ge 2 $. For each $ i\in \{ 1,\ldots ,k\} $, let $ G_i \cong K_4$ be the complete  graph  with four nodes $ r $, $ a_i $, $ b_{i,1} $ and $ b_{i,2}$, as
  shown in Figure~\ref{fig:auxiliary_grapg_G_i}.

        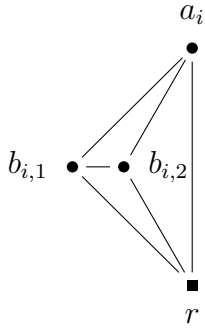
\begin{figure}
                \centering
                \begin{tikzpicture}[scale=0.9]
                        \node[fsquare,label=270:$ r $] (r) at (0,0) {};

                        \node[fcircle,label=90:$ a_i $] (a) at (0,3.5) {};
                        \node[fcircle,label=180:$ b_{i,1} $] (b1) at (-1.75,1.75) {};
                        \node[fcircle,label=0:$ b_{i,2} $] (b2) at (-1,1.75) {};

                        \draw (r) to (a) to (b1) to (r) to (b2) to (a);
                        \draw (b1) to (b2);
                \end{tikzpicture}
                \caption{Auxiliary graph $ G_i \cong K_4$}
                \label{fig:auxiliary_grapg_G_i}
        \end{figure}

        Now the graph $ G $ arises from the auxiliary graphs $ G_1,\ldots ,G_k $ by identifying the nodes $ r $.
        In addition, we insert edges $ \{ a_i,a_{i+1}\} $ for all $ i\in \{ 1,\ldots ,k-1\} $. $G$ is shown in Figure~\ref{fig:graph-G}.

        The node $ r $ is the root and $T := V(G)\setminus \{r\}$. The sinks $v\in T$ have the following delay weights:
        \begin{align*}
                w(v)
                := \begin{cases}
                  1 & \text{if } v \in \{ a_1,a_k\}, \\
                  0 & \text{if } v \in \{ a_2,\dots,a_{k-1}\}, \\
                  \frac{1}{2} + \frac{1}{k} & \text{if } v \in\{ b_{i,j} :  i \in \{1,\ldots,k\}, j\in \{1,2\}\}.
                \end{cases}
        \end{align*}

        Finally, we define the edge weights $ c\colon E(G)\to \bR _{\geq 0} $ by
        \begin{align*}
                c(\{ x,y\} )
                = \begin{cases}
                        1 & \text{if } r\in \{ x,y\} \text{ or } \{ x,y\} = \{ a_i, a_{i+1}\} \\
                        0 & \text{if } \{ x,y\} \in \{ \{ a_i, b_{i,1}\} , \{ a_i, b_{i,2}\} , \{ b_{i,1}, b_{i,2}\} \} .
                \end{cases}
        \end{align*}

        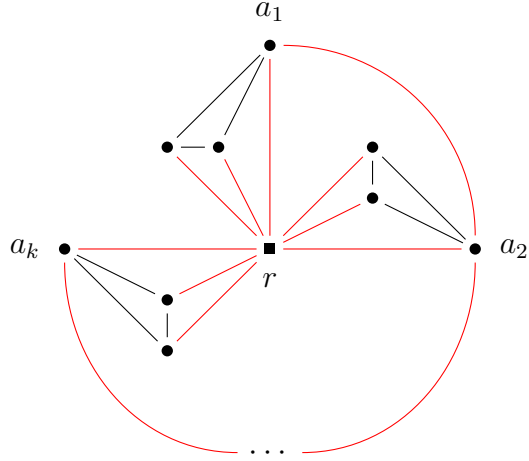
\begin{figure}[H]
                \centering
                \begin{tikzpicture}[scale=0.9]
                        \node[fsquare,label=270:$ r $] (r) at (0,0) {};

                        \node[fcircle,label=90:$ a_1 $] (a1) at (0,3) {};
                        \node[fcircle] (b11) at (-1.5,1.5) {};
                        \node[fcircle] (b12) at (-0.75,1.5) {};

                        \draw[color=red] (r) --  (a1);
                        \draw[color=red] (r) --  (b11);
                        \draw[color=red] (r) --  (b12);
                        \draw (a1) to (b11) ;
                        \draw  (b12) to (a1);
                        \draw (b11) to (b12);

                        \node[fcircle,label=0:$ a_2 $] (a2) at (3,0) {};
                        \node[fcircle] (b21) at (1.5,1.5) {};
                        \node[fcircle] (b22) at (1.5,0.75) {};

                        \draw[color=red] (r) --  (a2);
                        \draw[color=red] (r) --  (b21);
                        \draw[color=red] (r) --  (b22);

                        \draw (a2) to (b21) ;
                        \draw  (b22) to (a2);
                        \draw (b21) to (b22);

                        \node[fcircle,label=180:$ a_k $] (ak) at (-3,0) {};
                        \node[fcircle] (bk1) at (-1.5,-1.5) {};
                        \node[fcircle] (bk2) at (-1.5,-0.75) {};

                        \draw[color=red] (r) --  (ak);
                        \draw[color=red] (r) --  (bk1);
                        \draw[color=red] (r) --  (bk2);

                        \draw (ak) to (bk1) ;
                        \draw  (bk2) to (ak);
                        \draw (bk1) to (bk2);

                        \draw[color=red] (a1) to[out=0,in=90] (a2);

                        \node (tmp) at (0,-3) {$ \cdots $};

                        \draw[color=red] (a2) to[out=270,in=0] (tmp) to[out=180,in=270] (ak);
                \end{tikzpicture}
                \caption{The graph $ G $ that defines the metric. Red edges have cost 1, black edges have cost 0.}
                \label{fig:graph-G}
        \end{figure}

        We first determine the cost of a minimal Steiner tree. To do this, for each $ i\in \{ 1,\ldots ,k\} $ we contract the three nodes $ a_i $, $ b_{i,1} $ and $ b_{i,2} $. What remains is a graph with $ k + 1 $ nodes and minimum edge weight $ 1 $. Since each Steiner tree must contain $ k $ of the remaining edges, we have $ C_{SMT} \geq k $. Thus, the Steiner tree with edges $ \{ r,a_i\} $, $ \{ a_i, b_{i,2}\} $ and $ \{ b_{i,1} , b_{i,2}\} $ for all $ i\in \{ 1,\ldots ,k\} $ (see Figure \ref{fig91}) is optimum.

        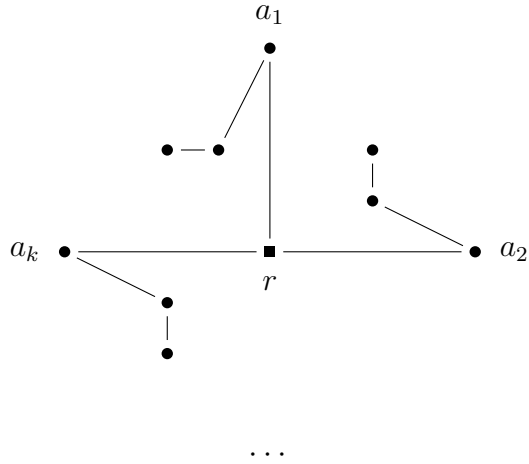
\begin{figure}
                \centering
                \begin{tikzpicture}[scale=0.9]
                        \node[fsquare,label=270:$ r $] (r) at (0,0) {};

                        \node[fcircle,label=90:$ a_1 $] (a1) at (0,3) {};
                        \node[fcircle] (b11) at (-1.5,1.5) {};
                        \node[fcircle] (b12) at (-0.75,1.5) {};

                        \draw (r) to (a1) to (b12) to (b11);

                        \node[fcircle,label=0:$ a_2 $] (a2) at (3,0) {};
                        \node[fcircle] (b21) at (1.5,1.5) {};
                        \node[fcircle] (b22) at (1.5,0.75) {};

                        \draw (r) to (a2) to (b22) to (b21);

                        \node[fcircle,label=180:$ a_k $] (ak) at (-3,0) {};
                        \node[fcircle] (bk1) at (-1.5,-1.5) {};
                        \node[fcircle] (bk2) at (-1.5,-0.75) {};

                        \draw (r) to (ak) to (bk2) to (bk1);

                        \node (tmp) at (0,-3) {$ \cdots $};
                \end{tikzpicture}
                \caption{An optimum (cost-distance) Steiner tree}
                \label{fig91}
        \end{figure}

        Moreover, for each sink terminal $ t \in T$, this tree contains a shortest $ r $-$ t $-path in $ (G,c) $, which is why Figure \ref{fig91} shows an optimum cost-distance Steiner tree. Its total cost is
        \begin{align}
                \label{eq89}
                k + \left( 2 + 2k\cdot \left( \frac{1}{2} + \frac{1}{k}\right) \right)
                = 2k + 4.
        \end{align}

In contrast Algorithm~\ref{alg:khazraei21} might compute $A_0$ as
the  minimum-length Steiner tree  shown in Figure~\ref{fig:min-length-tree}.

 \begin{figure}
   \centering
   \begin{tikzpicture}[scale=0.9]
     \node[fsquare,label=270:$ r $] (r) at (0,0) {};

     \node[fcircle,label=90:$ a_1 $] (a1) at (0,3) {};
                                        \node[fcircle] (b11) at (-1.5,1.5) {};
                                        \node[fcircle] (b12) at (-0.75,1.5) {};

                                        \draw (r) to (a1) to (b12) to (b11);

                                        \node[fcircle,label=0:$ a_2 $] (a2) at (3,0) {};
                                        \node[fcircle] (b21) at (1.5,1.5) {};
                                        \node[fcircle] (b22) at (1.5,0.75) {};

                                        \draw (a2) to (b22) to (b21);

                                        \node[fcircle,label=180:$ a_k $] (ak) at (-3,0) {};
                                        \node[fcircle] (bk1) at (-1.5,-1.5) {};
                                        \node[fcircle] (bk2) at (-1.5,-0.75) {};

                                        \draw (ak) to (bk2) to (bk1);

                                        \draw (a1) to[out=0,in=90] (a2);

                                        \node (tmp) at (0,-3) {$ \cdots $};

                                        \draw (a2) to[out=270,in=0] (tmp) to[out=180,in=270] (ak);
   \end{tikzpicture}
   \caption{$ \beta $-approximate minimum Steiner tree}
   \label{fig:min-length-tree}
 \end{figure}
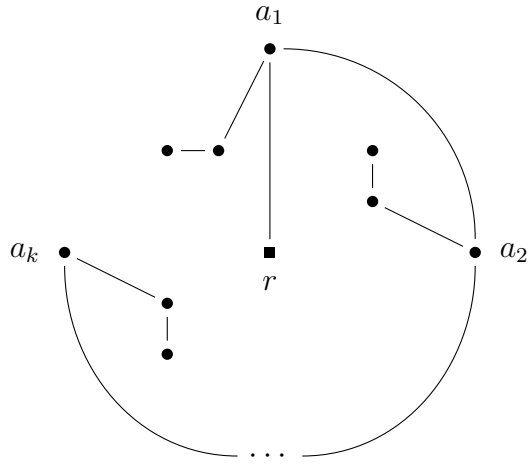

 This contains the edge $ \{ r,a_1\} $, for all $ i\in \{ 1,\ldots ,k-1\} $ the edges $ \{ a_i, a_{i+1}\} $ and for all $ i\in \{ 1,\ldots ,k\} $ the edges $ \{ a_i, b_{i,2}\} $, $ \{ b_{i,1}, b_{i,2}\} $.
Its length is  $ C = C_{SMT} = k $ and the minimum possible delay cost  is
\begin{align*}
  D
  = \sum _{t\in T} w(t) c(r,t)
  = 2 + 2k\cdot \left( \frac{1}{2} + \frac{1}{k}\right)
  = k + 4.
\end{align*}

Thus, the threshold $ \mu $ is chosen as
\begin{align*}
  \mu
  = \sqrt{\frac{D}{C}}
  = \sqrt{1 + \frac{4}{k}} .
\end{align*}

So  we obtain $ \mu > w(a_i) $, $ \mu > w(b_{i,1}) = w(b_{i,2}) $ and
\begin{align*}
  \mu
  < \sqrt{\left( 1 + \frac{2}{k}\right) ^2}
  = w(b_{i,1}) + w(b_{i,2}),
\end{align*}
for all $ i\in \{ 1,\ldots ,k\} $.

After transforming the Steiner tree into an arborescence with leaf set $T\cup\{r\}$ and the degree constraints, we obtain the $r$-arborescence
shown in Figure~\ref{fig:Steiner-arb-in-worst-case}. However, inserting these nodes does not change the remaining proof, so we identify the inserted Steiner vertices with their terminals again.

\begin{figure}
                            \centering
                                        \begin{tikzpicture}[scale=0.9]
                                                \node[fsquare,label=270:$ r $] (r) at (0,0) {};

                                                \node[hcircle] (a1) at (0,3) {};
                                                \node[hcircle] (a11) at (70:3) {};
                                                \node[fcircle,label=70:$ a_1 $] (a12) at (70:3.9) {};
                                                \node[fcircle] (b11) at (-1.5,1.5) {};
                                                \node[hcircle] (b12) at (-0.75,1.5) {};
                                                \node[fcircle] (b13) at (-0.75,0.75) {};

                                                \draw[->] (r) to (a1);
                                                \draw[->] (a1) to (b12);
                                                \draw[->] (b12) to (b11);
                                                \draw[->,blue] (b12) to (b13);

                                                \node[hcircle] (a2) at (3,0) {};
                                                \node[hcircle] (a21) at (340:3) {};
                                                \node[fcircle,label=340:$ a_2 $] (a22) at (340:3.9) {};
                                                \node[fcircle] (b21) at (1.5,1.5) {};
                                                \node[hcircle] (b22) at (1.5,0.75) {};
                                                \node[fcircle] (b23) at (0.75,0.75) {};

                                                \draw[->] (a2) to (b22);
                                                \draw[->] (b22) to (b21);
                                                \draw[->,blue] (b22) to (b23);

                                                \node[hcircle] (ak) at (-3,0) {};
                                                \node[fcircle,label=180:$ a_k $] (ak2) at (180:3.9) {};
                                                \node[fcircle] (bk1) at (-1.5,-1.5) {};
                                                \node[hcircle] (bk2) at (-1.5,-0.75) {};
                                                \node[fcircle] (bk3) at (-0.75,-0.75) {};

                                                \draw[->] (ak) to (bk2);
                                                \draw[->] (bk2) to (bk1);
                                                \draw[->,blue] (bk2) to (bk3);

                                                \draw[->,blue] (a1) to[out=0,in=160] (a11);
                                                \draw[->,blue] (a11) to (a12);
                                                \draw[->] (a11) to[out=340,in=90] (a2);

                                                \node (tmp) at (0,-3) {$ \cdots $};

                                                \draw[->,blue] (a2) to[out=270,in=70] (a21);
                                                \draw[->,blue] (a21) to (a22);
                                                \draw[->] (a21) to[out=250,in=0] (tmp);

                                                \draw[->,blue] (ak) to (ak2);
                                                \draw[->] (tmp) to[out=180,in=270] (ak);
                                        \end{tikzpicture}

                                        \caption{Steiner arborescence -  endpoints of blue edges are mapped to the same graph vertices.}
                                        \label{fig:Steiner-arb-in-worst-case}
                \end{figure}
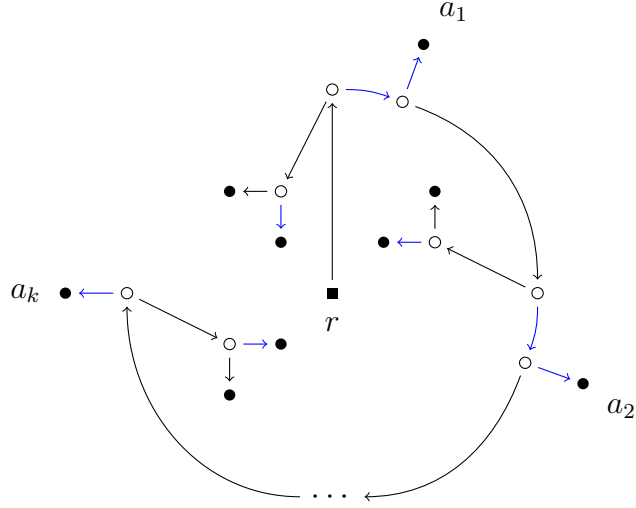

Regardless of the chosen topological order in Step 2 of Algorithm~\ref{alg:khazraei21}, the following properties hold:
\begin{enumerate}
\item No sub-arborescence with root $ b_{i,1} $ is ever cut off because the delay weight in the sub-arborescence is $ w(b_{i,1}) < \mu $.

\item A sub-arborescence with root $ b_{i,2} $ is always cut off, since the delay weight in the sub-arborescence is $ w(b_{i,1}) + w(b_{i,2}) > \mu $.

\item A sub-arborescence with root $ a_i $ is cut off exactly when $ i = 1 $, since according to the second property for $ i > 1 $ the delay weight in the sub-arborescence is $ w(a_k) < \mu $.
\end{enumerate}

Specifically, for each $ i\in \{ 1,\ldots ,k\} $ we remove the edge $ (a_i,b_{i,2}) $ and finally the edge $ (r,a_1) $. In this way, we thus obtain $ k + 1 $ cut off sub-arborescences.

  \begin{figure}
    \centering
    \begin{tikzpicture}[scale=0.9]
      \node[fsquare,label=270:$ r $] (r) at (0,0) {};

      \node[fcircle,label=90:$ a_1 $] (a1) at (0,3) {};
      \node[fcircle] (b11) at (-1.5,1.5) {};
      \node[fcircle] (b12) at (-0.75,1.5) {};

      \draw[->,red] (r) to (a1);
      \draw[->,red] (a1) to (b12);
      \draw[->] (b12) to (b11);

      \draw[dashed,blue] (-1.15,1.5) ellipse (0.8 and 0.4);

      \node[fcircle,label=0:$ a_2 $] (a2) at (3,0) {};
      \node[fcircle] (b21) at (1.5,1.5) {};
      \node[fcircle] (b22) at (1.5,0.75) {};

      \draw[->,red] (a2) to (b22);
      \draw[->] (b22) to (b21);

      \draw[dashed,blue] (1.5,1.15) ellipse (0.4 and 0.8);

      \node[fcircle,label=180:$ a_k $] (ak) at (-3,0) {};
      \node[fcircle] (bk1) at (-1.5,-1.5) {};
      \node[fcircle] (bk2) at (-1.5,-0.75) {};

      \draw[->,red] (ak) to (bk2);
      \draw[->] (bk2) to (bk1);

      \draw[dashed,blue] (-1.5,-1.15) ellipse (0.4 and 0.8);

      \draw[->] (a1) to[out=0,in=90] (a2);

      \node (tmp) at (0,-3) {$ \cdots $};

      \draw[->] (a2) to[out=270,in=0] (tmp) to[out=180,in=270] (ak);

      \draw[dashed,blue] (170:2.75) arc [start angle=170, delta angle=290, radius=2.75] -- (100:3.875) arc [start angle=100, delta angle=-290, radius=3.875] -- (170:2.75);
    \end{tikzpicture}

    \caption{Arborescence is divided into sub-arborescences outlined in blue. Red edges are removed in Step 2 of Algorithm~\ref{alg:khazraei21}.}
  \end{figure}
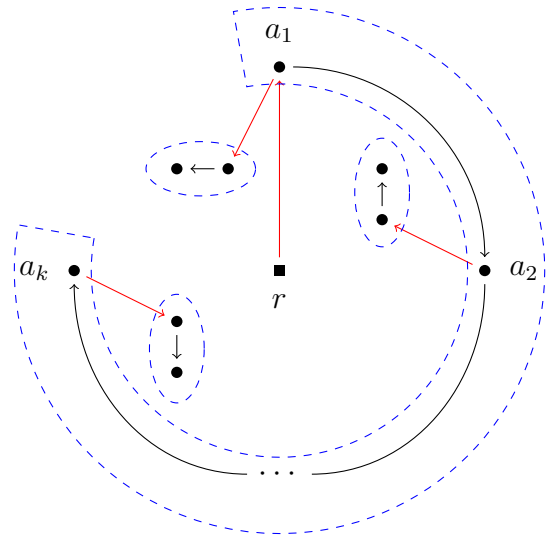

We immediately convince ourselves that for all arborescences it does not matter which terminal is used as a port for that component.
For each $ i\in \{ 1,\ldots ,k\} $, we choose the terminal $ b_{i,2} $ with the edge $ \{ r,b_{i,2}\} $. Furthermore we re-connect the terminal $ a_1 $ with the edge $ \{ r, a_1\} $.

        As a result, we get the following cost-distance Steiner tree:

        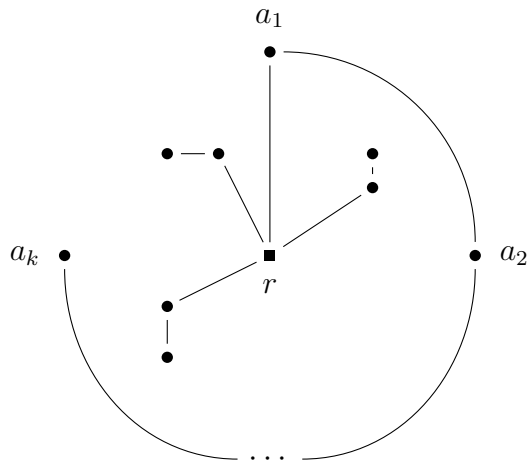
\begin{figure}[H]
                \centering
                \begin{tikzpicture}[scale=0.9]
                        \node[fsquare,label=270:$ r $] (r) at (0,0) {};

                        \node[fcircle,label=90:$ a_1 $] (a1) at (0,3) {};
                        \node[fcircle] (b11) at (-1.5,1.5) {};
                        \node[fcircle] (b12) at (-0.75,1.5) {};

                        \draw (r) to (b12) to (b11);

                        \node[fcircle,label=0:$ a_2 $] (a2) at (3,0) {};
                        \node[fcircle] (b21) at (1.5,1.5) {};
                        \node[fcircle] (b22) at (1.5,1) {};

                        \draw (r) to (b22) to (b21);

                        \node[fcircle,label=180:$ a_k $] (ak) at (-3,0) {};
                        \node[fcircle] (bk1) at (-1.5,-1.5) {};
                        \node[fcircle] (bk2) at (-1.5,-0.75) {};

                        \draw (r) to (bk2) to (bk1);

                        \draw (r) to (a1) to[out=0,in=90] (a2);

                        \node (tmp) at (0,-3) {$ \cdots $};

                        \draw (a2) to[out=270,in=0] (tmp) to[out=180,in=270] (ak);
                \end{tikzpicture}
                \caption{Computed cost-distance Steiner tree}
        \end{figure}

        The total cost of this solution is
        \begin{align}
                \label{eq88}
                (k + k) + \left( 1 + k + 2k\cdot \left( \frac{1}{2} + \frac{1}{k}\right) \right)
                = 4k + 3.
        \end{align}

Combining the results from (\ref{eq89}) and (\ref{eq88}) shows the asymptotic gap of $2$ for increasing $ k $.
\end{proof}

\begin{remark}
        We can use the same graph to formulate Theorem \ref{lem95} for any parameter $ \beta \leq 2 $. To do this, we generalize the delay weights
        \begin{align*}
                w(v)
                := \begin{cases}
                  \frac{1}{\beta}  & \text{if } v \in \{a_1,a_k\}, \\
                  0  & \text{if } v \in \{a_2,\dots,a_{k-1}\}, \\
                  \frac{1}{2} \frac{k}{1 + (k-1)\beta} + \frac{1}{1 + (k-1)\beta} & \text{if } v = b_{i,1} \text{ or } v = b_{i,2}
                \end{cases}
        \end{align*}

        and edge weights
        \begin{align*}
                c(\{ x,y\} )
                = \begin{cases}
                        1 & \text{if } r\in \{ x,y\} \\
                        \beta & \text{if } \{ x,y\} = \{ a_i, a_{i+1}\} \\
                        0 & \text{if } \{ x,y\} \in \{ \{ a_i, b_{i,1}\} , \{ a_i, b_{i,2}\} , \{ b_{i,1}, b_{i,2}\} \} .
                \end{cases}
        \end{align*}

        One easily computes that the instance constructed in this way allows the same choices as in the previous proof. The total cost of the solution calculated in this way deviates from the optimal objective function value by at least a factor $ 1 + \beta - \epsilon $ for  sufficiently large natural numbers $ k $.
\end{remark}

One could still hope that the modified algorithm provides a better approximation quality if the underlying metric space has more structure,
e.g.\ in the Manhattan plane.
However, we give a counterexample with a second proof of Theorem~\ref{lem95} in Appendix~\ref{appendix:tight-example-in-Manhattan-plane}.

\section{Improving the approximation ratio}
\label{sec:improvement}

When analyzing the worst-case example from Section~\ref{sec:worst_case_example}, it is noticeable that the cut-off arborescences can be classified into two groups:
\begin{enumerate}
        \item those with low delay weight (close to $ \mu $), low connection cost and high minimum possible delay cost and

        \item those with high delay weight (close to $ 2\mu $), high connection and low minimum possible delay cost.
\end{enumerate}

In this section, we improve the approximation quality by improving the
re-connection of arborescences in the second group and by
re-connecting the root component under certain circumstances.

\subsection{Splitting arborescences  into up to two components}
\label{sec:cut-into-two}
Assume  that we are given an arborescence $A'\in \cA$ with a high delay weight $W_{A'}$, a high connection cost $C_{A'}$, but a low minimum possible delay cost $D_{A'}$, e.g. as
as shown in Figure~\ref{fig:candidate-for-further-split} similar to the example in Section ~\ref{sec:worst_case_example}.
Then, it is reasonable to split the arborescence and connect the resulting sub-arborescences separately to the root so that paths to sinks with large delay weights are kept short as on the right  side.
\begin{figure}[t]
        \centering
        \begin{subfigure}{0.4\textwidth}
                \centering
                \begin{tikzpicture}
                        \node[fsquare,red,label=180:$ r $] (r) at (0,0) {};

                        \node[fcircle,label=90:\footnotesize{\textcolor{blue}{$ 1 $}}] (5) at (90:1.5) {};
                        \node[fcircle,label=45:\footnotesize{\textcolor{blue}{$ 0 $}}] (4) at (45:1.5) {};
                        \node[fcircle,label=0:\footnotesize{\textcolor{blue}{$ 0 $}}] (3) at (0:1.5) {};
                        \node[fcircle,label=315:\footnotesize{\textcolor{blue}{$ 0 $}}] (2) at (315:1.5) {};
                        \node[fcircle,label=270:\footnotesize{\textcolor{blue}{$ 0 $}}] (1) at (270:1.5) {};
                        \node[fcircle,label=225:\footnotesize{\textcolor{blue}{$ 1 $}}] (0) at (225:1.5) {};

                        \draw[->] (5) to (4);
                        \draw[->] (4) to (3);
                        \draw[->] (3) to (2);
                        \draw[->] (2) to (1);
                        \draw[->] (1) to (0);

                        \draw[dashed,->] (r) to (5);
                \end{tikzpicture}
                \caption{Cost: $ 6 + (1 + 6) = 13 $}
        \end{subfigure}
        \begin{subfigure}{0.4\textwidth}
                \centering
                \begin{tikzpicture}
                        \node[fsquare,red,label=180:$ r $] (r) at (0,0) {};

                        \node[fcircle,label=90:\footnotesize{\textcolor{blue}{$ 1 $}}] (5) at (90:1.5) {};
                        \node[fcircle,label=45:\footnotesize{\textcolor{blue}{$ 0 $}}] (4) at (45:1.5) {};
                        \node[fcircle,label=0:\footnotesize{\textcolor{blue}{$ 0 $}}] (3) at (0:1.5) {};
                        \node[fcircle,label=315:\footnotesize{\textcolor{blue}{$ 0 $}}] (2) at (315:1.5) {};
                        \node[fcircle,label=270:\footnotesize{\textcolor{blue}{$ 0 $}}] (1) at (270:1.5) {};
                        \node[fcircle,label=225:\footnotesize{\textcolor{blue}{$ 1 $}}] (0) at (225:1.5) {};

                        \draw[->] (5) to (4);
                        \draw[->] (4) to (3);
                        \draw[<-] (2) to (1);
                        \draw[<-] (1) to (0);

                        \draw[dashed,->] (r) to (5);
                        \draw[dashed,->] (r) to (0);
                \end{tikzpicture}
                \caption{Cost: $ 6 + (1 + 1) = 8 $}
        \end{subfigure}
        \caption{Connection of an arborescence before and after a split. $(M,c)$ is induced by  a complete graph with unit weights and delay weights are indicated by the blue node labels.}
        \label{fig:candidate-for-further-split}
\end{figure}
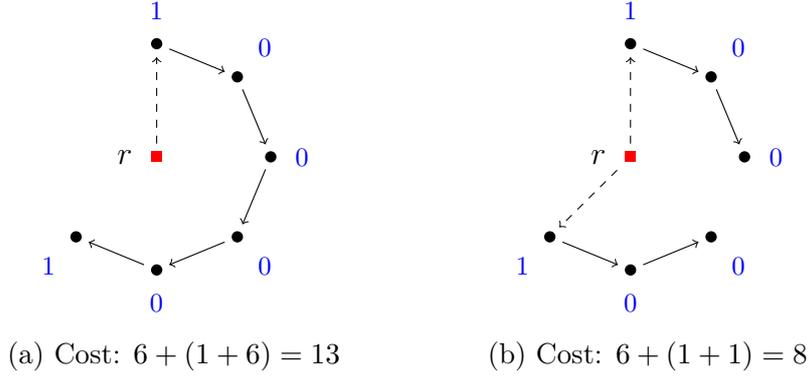

We first formalize this idea by  splitting an arborescence along a single  edge.
The following Lemma    reduces the coefficient of $\mu$  in the first  term of the upper bound from $1$ in Lemma~\ref{lem99} to $\frac{2}{3}$.
\begin{lem}
        \label{lem98}
        Let $ A'\in \cA $  be an arborescence with terminal set  $ T' $ that has been disconnected in Step 2 of Algorithm~\ref{alg:khazraei21}.
        Then the terminals from $ T' $ can be connected to the root $ r $ in $ \mathcal{O} (|E(A')| + |T_{A'}|) $ time with total cost at most
        $$
                \left( 1 + \frac{2}{3}\mu \right) C_{A'} + \left( 1 + \frac{1}{\mu}\right) D_{T'}.
                $$
\end{lem}
\begin{proof}

  If  $ T' $ consists of a  single terminal $ t $ with delay weight $ w(t) > \mu $, we connect it by a shortest $ r $-$ t $-path to  the root $ r $. This gives a total cost of at most
  $$    c(r,t) + w(t) c(r,t)
    \leq \frac{w(t)}{\mu} c(r,t) + D_{T'}
    = \left( 1 + \frac{1}{\mu}\right) D_{T'}.
  $$

  Otherwise, the root $v$ of $A'$ is a Steiner vertex with two outgoing edges $ (v,z_1),(v,z_2)\in \delta _{A'}^+ (v) $, since $ T' $  contains at least two terminals,
  and the deletion of only one edge would contradict the deletion of the edge entering $v$ in $A_0$.
  With  $W_{z_1} := W_{A'_{z_1}}$ and $W_{z_2} := W_{A'_{z_2}}$ this implies $ 0 < W_{z_1}, W_{z_2}\leq \mu $.

  \textbf{Case 1.}
$$    \frac{2\omega_{A'}}{W_{A'}} C_{A'} + \frac{1}{W_{A'}} D_{A'} \leq \frac{2}{3} \mu C_{A'} + \frac{1}{\mu} D_{A'}.
  $$
  We connect the sub-arborescence  as in Lemma \ref{lem99} and obtain a total cost of at most
  $$          \left( 1 + \frac{2\omega _{A'}}{W_{A'}}\right) C_{A'} + \left( 1 +  \frac{1}{W_{A'}}\right) D_{A'}
            \leq \left( 1 + \frac{2}{3} \mu \right) C_{A'} + \left( 1 + \frac{1}{\mu} \right) D_{A'}.
    $$

   \textbf{Case 2.}
                \begin{align}
                        \label{eq93}
                        \left( \frac{2\omega _{A'}}{W_{A'}} - \frac{2}{3} \mu \right) C_{A'}
                        > \left( \frac{1}{\mu} - \frac{1}{W_{A'}}\right) D_{A'}.
                \end{align}

                Since the sub-arborescence $ A' $ has been cut off, we have $ W_{A'} > \mu $ and obtain from (\ref{eq93}) an upper bound on the minimum possible  delay cost of $ A'$
                \begin{align}
                        \label{eq92}
                        D_{A'}
                        < \frac{2\omega _{A'} - \frac{2}{3} \mu W_{A'}}{W_{A'}} \cdot \frac{\mu W_{A'}}{W_{A'} - \mu} C_{A'}
                        = \frac{2}{3} \frac{3\omega _{A'} - \mu W_{A'}}{W_{A'} - \mu} \mu C_{A'}.
                \end{align}

                We split the arborescence $ A' $ along an edge $ e = (x,y)\in E(A') $ that attains the balance   $ \omega_{A'} = W_{A'_y} (W_{A'} - W_{A'_y}) $ as illustrated in
                Figure~\ref{fig:example-split} ($\omega_{A'}$ is defined in  Definition~\ref{def:omega_A}).

                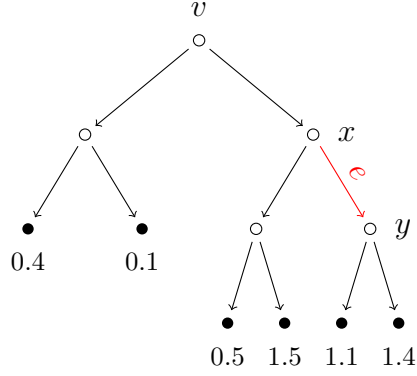
\begin{figure}[!t]
        \centering
        \begin{tikzpicture}
                \node[hcircle,label=90:$ v $] (v) at (0,0) {};
                \node[hcircle] (v1) at (-1.5,-1.25) {};
                \node[hcircle,label=0:$ x $] (v2) at (1.5,-1.25) {};
                \node[fcircle,label=270:\footnotesize{$ 0.4 $}] (v3) at (-2.25,-2.5) {};
                \node[fcircle,label=270:\footnotesize{$ 0.1 $}] (v4) at (-0.75,-2.5) {};
                \node[hcircle] (v5) at (0.75,-2.5) {};
                \node[hcircle,label=0:$ y $] (v6) at (2.25,-2.5) {};
                \node[fcircle,label=270:\footnotesize{$ 0.5 $}] (v7) at (0.375,-3.75) {};
                \node[fcircle,label=270:\footnotesize{$ 1.5 $}] (v8) at (1.125,-3.75) {};
                \node[fcircle,label=270:\footnotesize{$ 1.1 $}] (v9) at (1.875,-3.75) {};
                \node[fcircle,label=270:\footnotesize{$ 1.4 $}] (v10) at (2.625,-3.75) {};

                \draw[->] (v) to (v1);
                \draw[->] (v) to (v2);
                \draw[->] (v1) to (v3);
                \draw[->] (v1) to (v4);
                \draw[->] (v2) to (v5);
                \draw[->,red] (v2) to node[above,sloped] {$ e $} (v6);
                \draw[->] (v5) to (v7);
                \draw[->] (v5) to (v8);
                \draw[->] (v6) to (v9);
                \draw[->] (v6) to (v10);
        \end{tikzpicture}
        \caption{Exemplary illustration of the arborescence $ A' $ and an edge $ e = (x,y)\in E(A') $ with $ \omega _{A'} = W_{A'_y} (W_{A'} - W_{A'_y}) = 2.5 \cdot 2.5 = 6.25$.}
                \label{fig:example-split}
\end{figure}

                Let $W_y := W_{A'_y}$ and $W_x := W_{A'} - W_y$ be the total delay weights in the two resulting  components.
                Since the orientation of the edge $e$ does not matter for our further considerations, we may assume $ W_x \geq W_y $. Similarly, we assume $ W_{z_1}\geq W_{z_2}$. Then,
                \begin{align}\label{eqn:w_x-le-w_z1}
                        \begin{split}
                        \frac{W_{A'}^2}{4} - \left( W_{z_1} - \frac{W_{A'}}{2}\right)^2
                        &= W_{z_1} (W_{A'} - W_{z_1}) \\
                        &\leq \omega_{A'}=W_x(W_{A'} - W_x) \\
                        &= \frac{W_{A'}^2}{4} - \left( W_x - \frac{W_{A'}}{2}\right)^2.
                        \end{split}
                \end{align}

                As $ W_{A'} = W_x + W_y = W_{z_1} + W_{z_2} $ we also get $ W_x, W_{z_1}\geq \frac{W_{A'}}{2} $ and,   $ W_x \leq W_{z_1} $ according to  (\ref{eqn:w_x-le-w_z1}).
                Thus, we also have  $ W_y = W_{A'} - W_x  \geq W_{A'} - W_{z_1} =  W_{z_2} $.

                In the following we always consider the inequality chain
                $0
                        < W_y
                        \leq W_x
                        \leq \mu.$
                We connect the sub-arborescences $ A' - A'_y $ and $ A'_y $ separately to the root $ r $. We can use Lemma \ref{lem99} to obtain a total cost of at most
                \begin{align*}
                        &\hspace{-1em}\underbrace{\left( 1 + \frac{W_x}{2}\right) C_{A'-A'_y} + \left( 1 + \frac{1}{W_x}\right) D_{A'-A'_y}}_{\text{cost bound for } A' - A'_y } + \underbrace{\left( 1 + \frac{W_y}{2}\right) C_{A'_y} + \left( 1 + \frac{1}{W_y}\right) D_{A'_y}}_{\text{cost bound for } A'_y} \\
                        &\leq C_{A'} + D_{A'} + \frac{W_x}{2} C_{A'} + \frac{1}{W_y} D_{A'} \\
                        &= C_{A'} + \left( 1 + \frac{1}{\mu}\right) D_{A'} + \frac{W_x}{2} C_{A'} + \underbrace{\frac{\mu - W_y}{W_y \mu}}_{\geq 0} D_{A'} \\
                        \overset{(\ref{eq92})}&{\leq} \left( 1 + \frac{W_x}{2} + \frac{2}{3} \frac{\mu - W_y}{W_y} \cdot \frac{3\omega _{A'} - \mu W_{A'}}{W_{A'} - \mu}\right) C_{A'} + \left( 1 + \frac{1}{\mu}\right) D_{A'}.
                \end{align*}

                To bound this term we need to find the supremum of  the function $ f\colon (0,\mu ]^2\to \bR $ defined by
                  $$
                  f(x,y)
                        := \frac{x}{2} + \frac{2}{3} \frac{\mu - y}{y}\cdot \frac{3xy - \mu (x+y)}{x+y-\mu}
                  $$
                under the constraints $ x + y > \mu $ and $ x \geq y $. To this end, we first note that
                \begin{align*}
                        \frac{\partial f}{\partial x} (x,y)
                        &= \frac{1}{2} + \frac{2}{3} \frac{\mu - y}{y}\cdot \frac{(3y-\mu )(x+y-\mu ) - (3xy - \mu (x+y))\cdot 1}{(x+y-\mu )^2} \\
                        &= \frac{1}{2} + \frac{2}{3} \frac{\mu - y}{y}\cdot \frac{3y^2 - 3\mu y + \mu ^2}{(x+y-\mu )^2} \\
                        &= \frac{1}{2} + \underbrace{\frac{\mu - y}{y}}_{\geq 0} \cdot \underbrace{\frac{(\sqrt{2} y - \frac{1}{\sqrt{2}} \mu )^2 + \frac{1}{6} \mu ^2}{(x+y-\mu )^2}}_{\geq 0}> 0.
                \end{align*}
                Thus,  $ f(\cdot ,y) $ is strictly monotonically increasing for each $ y\in (0,\mu ] $, and  $ f(x,y)\leq f(\mu ,y) =: g(y) $ satisfying
                \begin{align*}
                        g(y)
                        &= \frac{\mu}{2} + \frac{2}{3} \frac{\mu - y}{y} \cdot \frac{2\mu y - \mu ^2}{y}
                        = \frac{\mu}{2} + \frac{2}{3} \left( \frac{3}{y} \mu ^2 - \frac{\mu ^3}{y^2} - 2\mu \right)
                        = -\frac{5}{6} \mu - \frac{2\mu ^3}{3y^2} + \frac{2\mu ^2}{y}.
                \end{align*}

                A simple analysis of its domain boundary as well as first and second order conditions shows that $g$ attains the maximum  at $ y = \frac{2}{3} \mu $  with value
                $$
                        g\left( \frac{2}{3} \mu \right)
                        = - \frac{5}{6} \mu - \frac{3}{2} \mu + 3 \mu
                        = \frac{2}{3} \mu.
                $$
        Now we take a look at the running time: Note, that during Step 2 of Algorithm \ref{alg:khazraei21}, we can also store with every vertex $v$ the total delay weight in the sub-arborescence rooted at $ v $.

        We can determine in $ \mathcal{O} (|E(A')|+|T_{A'}|) $ time the delay weight $ W_{A'} $, the connection cost $ C_{A'} $, the minimum possible delay cost $ D_{A'} $, and the balance $ \omega _{A'} $ of $ A' $. For $ \omega _{A'} $, we use the total delay weight of each vertex' sub-arborescence computed during Step 2.

        In the same run time we can also compute these values for the sub-arborescences $ A'_y $ and $ A'-A'_y $ and find out the underlying case. We just need to make sure that we can find the ``port'' vertex quickly enough when using Lemma \ref{lem99}. Specifically we find an optimal ``port'' vertex in $ A' $ in $ \mathcal{O} (|E(A')| + |T_{A'}|) $ time as follows:

        For each node $ v\in V(A') $, the cost when using $ v $ as the port is
        \begin{align*}
                \cost_v
                = c(r,v) + C_A'  + \sum _{t\in T'} w(t)\cdot (c(r,v) +  c(E(A'_{[v,t]}))).
        \end{align*}

        So for an edge $ e = (x,y)\in E(A') $ we have
        \begin{align*}
                \cost _x - \cost_y
                &= c(r,x) - c(r,y) + \sum _{t\in T'} w(t) (c(r,x) - c(r,y)) \\
                &\quad + \sum _{t\in T_{A_y}} w(t) \underbrace{(c(E(A'_{[x,t]})) - c(E(A'_{[y,t]})))}_{= c(e)} \\
                &\quad + \sum _{t\in T_{A'}\backslash T_{A_y}} w(t) \underbrace{(c(E(A'_{[x,t]})) - c(E(A'_{[y,t]})))}_{= -c(e)} \\
                &= (c(r,x) - c(r,y)) (1+W_{A'}) + c(e) W_{A_y} - c(e) (W_{A'} - W_{A_y}) \\
                &= (c(r,x) - c(r,y)) (1+W_{A'}) + c(e) (2W_{A_y} - W_{A'}),
        \end{align*}

        This allows us to compute in constant time the  cost for choosing  $ y $ as the port from  the cost for choosing its parent $ x $ as the port. We take advantage of this property and first compute the  cost for using the root of $ A' $ as the port in $ \mathcal{O} (|E(A')| + |T_{A'}|) $. Then, we find the  find the best ``port'' vertex in a top-down traversal in the claimed linear time.

        Using this technique also for the sub-arborescences $ A'_y $ and $ A'-A'_y $ that may occur finishes the proof.

\end{proof}
\setcounter{case}{0}

Note that \cite{Khazraei21} required a running time of
$\mathcal{O}(|T'|^2)$ to select a port vertex.

\subsection{Improving the root component $A_r$}
\label{sec:improve-root-component}
In \cite{Khazraei21} the   arborescence $ A_r $, which   was not cut off in Step 2 of Algorithm~\ref{alg:khazraei21},
was kept unaltered. We  show how to connect it  in a better way.

\begin{lem}
        \label{lem97}
        Let $ A_r $ be the sub-arborescence of $ A_0 $ rooted at $ r $ after Step 2  of Algorithm \ref{alg:khazraei21}.
        The terminal set $ T_{A_r} $ can be connected to the root $ r $ in $ \mathcal{O} (|E(A_r)| + |T_{A_r}|) $ time with total cost at most
        \begin{align*}
                \left( 1 + \frac{\mu}{2} \right) C_{A_r} + \left( 1 + \frac{1}{\mu}\right) D_{A_r}
        \end{align*}
\end{lem}
\begin{proof}
        Let $ (r,x)\in \delta _{A_r}^+ (r) $ be arbitrary and $ A_x $ the arborescence of $ A_r-r $ rooted at $ x $. We show that the terminal set $ T_{A_x} $ can be connected to the root $ r $ with total cost at most
        \begin{align*}
                \left( 1 + \frac{\mu}{2}\right) (C_{A_x} + c(r,x)) + \left( 1 + \frac{1}{\mu}\right) D_{A_x}.
        \end{align*}

        Adding this cost for all edges in $ \delta _{A_r}^+ (r) $, we obtain the claim.

        We distinguish again between two cases:

        \textbf{Case 1.}
              $$
                        W_{A_x} (C_{A_x} + c(r,x))
                        \leq \frac{\mu}{2} (C_{A_x} + c(r,x)) + \frac{1}{\mu} D_{A_x}.
                $$
                By keeping the arborescence $ A_x $ connected through $ (r,x) $, the connection cost is $ C_{A_x} + c(r,x) $. In particular, for each terminal $ t\in T_{A_x} $, the $ r $-$ t $-path in $ A_x + (r,x) $ has a length of at most $ C_{A_x} + c(r,x) $. We therefore obtain a total cost of at most
                $$
                        (1 + W_{A_x}) (C_{A_x} + c(r,x))
                        \leq \left( 1 + \frac{\mu}{2} \right) (C_{A_x} + c(r,x)) + \frac{1}{\mu} D_{A_x} .
                $$
        \textbf{Case 2.}
                \begin{align}
                        \label{eq91x}
                        W_{A_x} (C_{A_x} + c(r,x))
                        > \frac{\mu}{2} (C_{A_x} + c(r,x)) + \frac{1}{\mu} D_{A_x}.
                \end{align}

                                Therefore we have $ W_{A_x} > 0 $ and obtain from (\ref{eq91x}) an upper bound on the minimum possible delay cost of $ A_x $
                \begin{align}
                        \label{eq90x}
                        D_{A_x}
                        < \left( W_{A_x} - \frac{\mu}{2}\right) \mu (C_{A_x} + c(r,x)).
                \end{align}

                We remove the edge $ (r,x) $ and connect the arborescence $ A_x $ to the root $ r $ as in Lemma \ref{lem99} and obtain total cost of at most
                \begin{align*}
                        &\hspace{-1em}\left( 1 + \frac{W_{A_x}}{2}\right) C_{A_x} + \left( 1 + \frac{1}{W_{A_x}}\right) D_{A_x} \\
                        &= \left( 1 + \frac{W_{A_x}}{2}\right) C_{A_x} + \left( 1 + \frac{1}{\mu}\right) D_{A_x} + \underbrace{\frac{\mu - W_{A_x}}{\mu W_{A_x}}}_{\geq 0} D_{A_x} \\
                        \overset{(\ref{eq90x})}&{\leq} \left( 1 + \frac{W_{A_x}}{2}\right) C_{A_x} + \left( 1 + \frac{1}{\mu}\right) D_{A_x} + \frac{\mu - W_{A_x}}{\mu W_{A_x}} \left( W_{A_x} - \frac{\mu}{2}\right) \mu (C_{A_x} + c(r,x)) \\
                        &\leq \left( 1 - \frac{W_{A_x}}{2} + \frac{3}{2} \mu - \frac{\mu ^2}{2 W_{A_x}}\right) (C_{A_x} + c(r,x)) + \left( 1 + \frac{1}{\mu}\right) D_{A_x}.
                \end{align*}

                So we need to bound the supremum of  $ f\colon (0 , \mu ]\to \bR $ defined by
                \begin{align*}
                        f(x)
                        := -\frac{x}{2} + \frac{3}{2} \mu - \frac{\mu ^2}{2x} .
                \end{align*}

                Notice that $ f $ is monotonically increasing
                \begin{align*}
                        f'(x)
                        = -\frac{1}{2} + \frac{\mu ^2}{2x^2}
                        \geq -\frac{1}{2} + \frac{1}{2}
                        = 0.
                \end{align*}

                Therefore, we obtain $f(x)
                        \leq f(\mu )
                        = -\frac{\mu}{2} + \frac{3}{2}\mu - \frac{\mu}{2}
                        = \frac{\mu}{2}$, proving the claimed bound.

        For the running time analysis we use the same technique from the proof of Lemma \ref{lem98} with the advantage that here we only need to determine quantities for the single arborescence $ A_r $.
\end{proof}
\setcounter{case}{0}

\subsection{Further improvements}
\label{sec:cut-into-three}

Using Lemmas \ref{lem98} and \ref{lem97}, we can already show a better approximation quality than $1 + \beta$.
For now, let us extend the idea of splitting an arborescence along an edge to multiple edges.

In Case 2 of the proof of Lemma~\ref{lem98} we have used   $ W_x\geq W_y $ to prove
\begin{align}
        \frac{W_x}{2} C_{A'-A'_y} + \frac{1}{W_x} D_{A'-A'_y} + \frac{W_y}{2} C_{A'_y} + \frac{1}{W_y} D_{A'_y}
        \leq \frac{W_x}{2} C_{A'} + \frac{1}{W_y} D_{A'}.
        \label{eqn:tight-inequality}
\end{align}

However, at the end of the proof we found that the cost in Case 2 becomes largest for $ W_x = \mu $ and $ W_y = \frac{2}{3} \mu $. The inequality (\ref{eqn:tight-inequality})
can only be tight if  $ C_{A'-A'_y} = C_{A'} $ and $ D_{A'-A'_y} = 0 $.
Thus, the arborescence $ A'-A'_y $ would have a high delay weight, a high connection cost and low minimum possible delay cost.
$ A'-A'_y $ has the same characteristics that we discussed at the beginning of Section~\ref{sec:cut-into-two}.

Therefore, the arborescence $ A'-A'_y $ should again be split under certain conditions at an edge $(x_1,x_2)\in E(A'-A'_y)$, that attains the
balance $\omega_{A'-A'_y}$.
This way,  we obtain a total of up to three sub-arborescences, which are all connected separately in Step 3 of Algorithm~\ref{alg:khazraei21}.
In our analysis of this approach we will use the following four functions.
\begin{definition}
  Let $ b \ge 0 $, $ \mu > 0 $, and $X^{\mu} := \{ (x,y) \in (0,\mu]^2: x+ y - \mu > 0\}$.
    We define the functions $ \phi _0, \phi _1, \phi _2: X^{\mu} \to \mathbb{R}$ and $\phi _3: X^{\mu} \times (0,\mu]\to \mathbb{R} $  as
        \begin{align*}
                \phi _0 (x,y)
                &:= \frac{\mu - y}{y}\cdot \frac{2xy - b\mu (x+y)}{x+y - \mu} \\
                \phi _1 (x,y)
                &:= \frac{y}{2} + \phi _0 (x,y) \\
                \phi _2 (x,y)
                &:= \frac{3}{8} x + \phi _0 (x,y) \\
                \phi _3 (x,x_2,y)
                &:= \frac{3}{2} x_2 - \frac{x_2 (4x_2 - x)}{2y} + \phi _0 (x,y).
        \end{align*}
\end{definition}

In Lemma~\ref{lem98}, the first  part of the cost bound was $(1 + \frac{2}{3}\mu)C_{A'}$.
We try to replace the coefficient $\frac{2}{3}$ by a smaller number $b$.
First, we derive some sufficient conditions for $b$.

\begin{lem}
        Let $ b\in \bR _{\geq 0} $ be a constant that satisfies for all $\mu > 0$
        \begin{enumerate}
                \item $ \phi _1 (x,y) \leq b\mu $ for all  $ (x,y)\in X^{\mu}$,

                \item $ \phi _2 (x,y) \leq b\mu $ for all  $ (x,y)\in X^{\mu}$, and

                \item $ \phi _3 (x,x_2,y) \leq b\mu $ for all $ (x,y) \in X^{\mu}, x_2 \in (0,\mu ] $.
        \end{enumerate}

        Furthermore, let $ A'\in \cA $ be an arborescence that was cut off in Step 2 of Algorithm~\ref{alg:khazraei21}. Then the corresponding terminals in $ T' $ can be connected to the root $ r $ in $ \mathcal{O} (|E(A')| + |T_{A'}|) $ time with total cost at most
        \begin{align*}
                \left( 1 + b \mu \right) C_{A'} + \left( 1 + \frac{1}{\mu}\right) D_{A'}.
        \end{align*}
\end{lem}

\begin{proof}
        If the arborescence $ A' $ contains only one terminal $ t\in T_{A'} $ with positive delay weight, then the terminals $ T_{A'} $ can be connected to the root $ r $ with a total cost of at most
        \begin{align*}
                C_{A'} + \left( 1 + \frac{1}{W_{A'}}\right) D_{A'}
                \leq (1 + b\mu )C_{A'} + \left( 1 + \frac{1}{\mu}\right) D_{A'}.
        \end{align*}

        In the following, we therefore assume that the arborescence $ A' $ contains at least two terminals with positive delay weight. We now distinguish the following cases:

        \begin{case}
          $$
                        \frac{2\omega _{A'}}{W_{A'}} C_{A'} + \frac{1}{W_{A'}} D_{A'}
                        \leq b\mu C_{A'} + \frac{1}{\mu} D_{A'}.
                        $$

                Then we connect the arborescence $ A' $ as in Lemma \ref{lem99} to the root $ r $ and obtain a total cost of at most
                \begin{align*}
                        \left( 1 + \frac{2\omega _{A'}}{W_{A'}}\right) C_{A'} + \left( 1 +  \frac{1}{W_{A'}}\right) D_{A'}
                        \leq \left( 1 + b\mu \right) C_{A'} + \left( 1 + \frac{1}{\mu} \right) D_{A'}.
                \end{align*}
        \end{case}

        \begin{case}
                \begin{align}
                        \label{eq87}
                        \left( \frac{2\omega _{A'}}{W_{A'}} - b\mu \right) C_{A'}
                        > \left( \frac{1}{\mu} - \frac{1}{W_{A'}}\right) D_{A'}.
                \end{align}

                Since the sub-arborescence $ A' $ has been cut off, we have $ W_{A'} > \mu $ and we obtain from (\ref{eq87}) an upper bound on the minimum possible delay cost of $ A' $
                \begin{align}
                        \label{eq85}
                        D_{A'}
                        < \frac{2\omega _{A'} - b\mu W_{A'}}{W_{A'}} \cdot \frac{\mu W_{A'}}{W_{A'} - \mu} C_{A'}
                        = \frac{2\omega _{A'} - b\mu W_{A'}}{W_{A'} - \mu} \mu C_{A'}.
                \end{align}

                We split the arborescence $ A' $ at an edge $ e = (x,y)\in E(A') $ so that the resulting arborescences $A'- A'_y $ and $ A'_y  $ satisfy the equation $\omega _{A'} = W_{(A'-A'_y)} W_{A_y} $.

                To shorten our notation, we define $A_y := A'_y$ and $A_x := A'-A_y$, to be the two arborescences in $A'-e$ containing $y$ and $x$ (though $A_x$ need not be rooted at $x$).

                Since the orientation of the edge $ e $ does not matter for our further considerations, we can assume $ W_{A_x}\geq W_{A_y} $ without loss of generality. As discussed in the proof of Lemma \ref{lem98}, the following bounds hold:
                \begin{itemize}
                        \item $ W_{A_x}, W_{A_y} > 0 $ due to  $|V(A')\cap T|\ge 2$ and the choice of $e$, and

                        \item $ W_{A_x}, W_{A_y}\leq \mu $, since otherwise the arborescences $ A_x $ and $ A_y $  would have been cut off earlier.
                \end{itemize}

                If the arborescence $ A_x $ again contains only one terminal with positive delay weight, we connect the arborescences $ A_x $ and $ A_y $ respectively as in Lemma \ref{lem99} to the root $ r $.
                This  yields a total cost of at most
                \begin{align*}
                        &\hspace{-1em} \underbrace{C_{A_x} + \left( 1 + \frac{1}{W_{A_x}}\right) D_{A_x}}_{\text{cost bound for } A_x} + \underbrace{\left( 1 + \frac{W_{A_y}}{2}\right) C_{A_y} + \left( 1 + \frac{1}{W_{A_y}}\right) D_{A_y}}_{\text{cost bound for } A_y} \\
                        &\leq \left( 1 + \frac{W_{A_y}}{2}\right) C_{A'} + \left( 1 + \frac{1}{W_{A_y}}\right) D_{A'} \\
                        &= \left( 1 + \frac{W_{A_y}}{2}\right) C_{A'} + \left( 1 + \frac{1}{\mu}\right) D_{A'} + \underbrace{\frac{\mu - W_{A_y}}{\mu W_{A_y}}}_{\geq 0} D_{A'} \\
                        \overset{(\ref{eq85})}&{\leq} \Bigl( 1 + \underbrace{\frac{W_{A_y}}{2} + \frac{\mu - W_{A_y}}{W_{A_y}}\cdot \frac{2\omega _{A'} - b\mu W_{A'}}{W_{A'} - \mu}}_{= \phi _1 (W_{A_x}, W_{A_y})} \Bigr) C_{A'} + \left( 1 + \frac{1}{\mu}\right) D_{A'} \\
                        &\leq (1 + b\mu ) C_{A'} + \left( 1 + \frac{1}{\mu}\right) D_{A'}.
                \end{align*}

                So we assume that the arborescence $ A_x $ also contains at least two terminals with positive delay weight.
                This implies $ \omega _{A_x} > 0$. Let $ e' = (x_1, x_2)\in E(A_x) $ be an edge that attains the positive balance $ \omega _{A_x}$.

                Again, for a shorter notation we define $A_{x_2} := (A_x)_{x_2}$ and $A_{x_1} := (A_x - (A_x)_{x_2})$ to be the
                two arborescences in $ A_x - e' $ containing $x_2$ and $x_1$.
                By the choice of $e'$, we have $ \omega _{A_x}  = W_{A_{x_1}}\cdot W_{A_{x_2}}$.
                As $ \omega _{A_x} > 0$, $ W_{A_{x_1}}, W_{A_{x_2}} > 0 $.
                Without loss of generality  $ W_{A_{x_1}} \geq W_{A_{x_2}} $.

                In the third of the following cases, we  will further split the arborescence $A_x$ into $A_{x_1}$ and $A_{x_2}$:
                \begin{subcase}
                  $$ W_{A_{x_1}} \geq 3 W_{A_{x_2}}. $$

                  This implies $ W_{A_x} = W_{x_1}+W_{x_2} \geq 4 W_{A_{x_2}}$.
                  We obtain a better upper bound for $ \frac{2\omega _{A_x}}{W_{A_x}} $ than $ \frac{W_{A_x}}{2} $, because
                        \begin{align*}
                                \frac{2\omega _{A_x}}{W_{A_x}}
                                &= \frac{2 W_{A_{x_1}} W_{A_{x_2}}}{W_{A_x}} \\
                                &= \frac{W_{A_x}}{2} - \frac{2\left( W_{A_{x_1}} - \frac{W_{A_x}}{2}\right) \left( \frac{W_{A_x}}{2} - W_{A_{x_2}}\right)}{W_{A_x}} \\
                                &= \frac{W_{A_x}}{2} - \frac{2\left( \frac{W_{A_x}}{2} - W_{A_{x_2}}\right) ^2}{W_{A_x}} \\
                                &\leq \frac{W_{A_x}}{2} - \frac{2\left( \frac{W_{A_x}}{2} - \frac{W_{A_x}}{4}\right) ^2}{W_{A_x}} \\
                                &= \frac{3}{8} W_{A_x}.
                        \end{align*}

                        By connecting the arborescences $ A_x $ and $ A_y $ separately as in Lemma \ref{lem99} to the root $ r $, we obtain a total cost of at most
                        \begin{align*}
                                &\hspace{-1em} \underbrace{\left( 1 + \frac{2\omega _{A_x}}{W_{A_x}}\right) C_{A_x} + \left( 1 + \frac{1}{W_{A_x}}\right) D_{A_x}}_{\text{cost bound for } A_x} + \underbrace{\left( 1 + \frac{W_{A_y}}{2}\right) C_{A_y} + \left( 1 + \frac{1}{W_{A_y}}\right) D_{A_y}}_{\text{cost bound for } A_y} \\
                                &\leq \left( 1 + \frac{3}{8} W_{A_x}\right) C_{A_x} + \left( 1 + \frac{W_{A_y}}{2}\right) C_{A_y} + \left( 1 + \frac{1}{W_{A_y}}\right) D_{A_v} \\
                                &\leq \left( 1 + \max \left\{ \frac{3}{8} W_{A_x}, \frac{W_{A_y}}{2}\right\} \right) C_{A'} + \left( 1 + \frac{1}{\mu}\right) D_{A'} + \underbrace{\frac{\mu - W_{A_y}}{\mu W_{A_y}}}_{\geq 0} D_{A'} \\
                                \overset{(\ref{eq85})}&{\leq} \Bigl( 1 + \underbrace{\max \left\{ \frac{3}{8} W_{A_x}, \frac{W_{A_y}}{2}\right\} + \frac{\mu - W_{A_y}}{W_{A_y}}\cdot \frac{2\omega _{A'} - b\mu W_{A'}}{W_{A'} - \mu}}_{= \max \{ \phi _2(W_{A_x}, W_{A_y}), \phi _1 (W_{A_x},W_{A_y})\}}\Bigr) C_{A'} + \left( 1 + \frac{1}{\mu}\right) D_{A'} \\
                                &\leq ( 1 + b\mu ) C_{A'} + \left( 1 + \frac{1}{\mu}\right) D_{A'}.
                        \end{align*}
                \end{subcase}

                \begin{subcase}
                  $ W_{A_{x_1}} < 3 W_{A_{x_2}} $ \textbf{and}
                        \begin{align}
                                \label{eq84}
                                \left( \frac{2\omega _{A_x}}{W_{A_x}} - \frac{W_{A_{x_1}}}{2}\right) C_{A_x}
                                < \left( \frac{1}{W_{A_{x_2}}} - \frac{1}{W_{A_x}}\right) D_{A_x} .
                        \end{align}

                        Then,
                        \begin{align*}
                                \frac{2\omega _{A_x}}{W_{A_x}} - \frac{W_{A_{x_1}}}{2}
                                = \frac{4 W_{A_{x_1}} W_{A_{x_2}} - W_{A_{x_1}} (W_{A_{x_1}} + W_{A_{x_2}})}{2W_{A_x}}
                                = \frac{W_{A_{x_1}} (3 W_{A_{x_2}} - W_{A_{x_1}})}{2 W_{A_x}}
                                > 0.
                        \end{align*}

                        Thus, (\ref{eq84}) provides an upper bound on the connection cost of $ A_x $
                        \begin{align}
                                \label{eq83}
                                C_{A_x}
                                < \frac{W_{A_x} - W_{A_{x_2}}}{W_{A_{x_2}} W_{A_x}} \cdot \frac{2 W_{A_x}}{W_{A_{x_1}} (3 W_{A_{x_2}} - W_{A_{x_1}})} D_{A_x}
                                = \frac{2}{W_{A_{x_2}} (3 W_{A_{x_2}} - W_{A_{x_1}})} D_{A_x}.
                        \end{align}

                        We set $ \alpha := \left( \frac{1}{W_{A_y}} - \frac{1}{W_{A_x}}\right) \frac{W_{A_{x_2}} (3 W_{A_{x_2}} - W_{A_{x_1}})}{2}\geq 0 $ and get
                        \begin{align}
                                \begin{split}
                                        \frac{2\omega _{A_x}}{W_{A_x}} - \alpha
                                        &= \frac{2W_{A_{x_1}} W_{A_{x_2}}}{W_{A_x}} + \left( \frac{1}{W_{A_x}} - \frac{1}{W_{A_y}}\right) \frac{W_{A_{x_2}} (3 W_{A_{x_2}} - W_{A_{x_1}})}{2} \\
                                        &= \frac{3W_{A_{x_1}} W_{A_{x_2}} + 3 W_{A_{x_2}}^2}{2 W_{A_x}} - \frac{W_{A_{x_2}} (4 W_{A_{x_2}} - W_{A_x})}{2 W_{A_y}} \\
                                        &= \frac{3}{2} W_{A_{x_2}} - \frac{W_{A_{x_2}} (4 W_{A_{x_2}} - W_{A_x})}{2 W_{A_y}} .
                                \end{split}
                        \end{align}

                        Connecting $ A_x $ and $ A_y $ to the root $ r $ separately as in Lemma \ref{lem99}, we obtain total costs of at most
                        \begin{align*}
                                &\hspace{-1em} \underbrace{\left( 1 + \frac{2\omega _{A_x}}{W_{A_x}}\right) C_{A_x} + \left( 1 + \frac{1}{W_{A_x}}\right) D_{A_x}}_{\text{cost bound for } A_x} + \underbrace{\left( 1 + \frac{W_{A_y}}{2}\right) C_{A_y} + \left( 1 + \frac{1}{W_{A_y}}\right) D_{A_y}}_{\text{cost bound for } A_y} \\
                                \overset{(\ref{eq83})}&{<}\!\! \left( 1 + \frac{2\omega _{A_x}}{W_{A_x}} - \alpha \right)\! C_{A_x}\! + \!\left( 1 + \frac{1}{W_{A_y}}\right)\! D_{A_x}\! + \!\left( 1 + \frac{W_{A_y}}{2}\right)\! C_{A_y}\! + \!\left( 1 + \frac{1}{W_{A_y}}\right)\! D_{A_y} \\
                                &\leq \left( 1 + \max \left\{ \frac{2\omega _{A_x}}{W_{A_x}} - \alpha , \frac{W_{A_y}}{2}\right\} \right) C_{A'} + \left( 1 + \frac{1}{\mu}\right) D_{A'} + \underbrace{\frac{\mu - W_{A_y}}{\mu W_{A_y}}}_{\geq 0} D_{A'} \\
                                \overset{(\ref{eq85})}&{\leq} \Bigl( 1 + \underbrace{\max \left\{ \frac{2\omega _{A_x}}{W_{A_x}} - \alpha , \frac{W_{A_y}}{2}\right\} + \frac{\mu - W_{A_y}}{W_{A_y}}\cdot \frac{2\omega _{A'} - b\mu W_{A'}}{W_{A'} - \mu}}_{ \overset{(\ref{eq83})}{=} \max \{ \phi _3 (W_{A_x}, W_{A_{x_2}}, W_{A_y}), \phi _1 (W_{A_x}, W_{A_y})\} } \Bigr) C_{A'} + \left( 1 + \frac{1}{\mu}\right) D_{A'} \\
                                &\leq (1 + b\mu ) C_{A'} + \left( 1 + \frac{1}{\mu}\right) D_{A'}.
                        \end{align*}
                \end{subcase}

                \begin{subcase}
                         $ W_{A_{x_1}} < 3 W_{A_{x_2}} $ \textbf{and}
                        \begin{align*}
                                \left( \frac{2\omega _{A_x}}{W_{A_x}} - \frac{W_{A_{x_1}}}{2}\right) C_{A_x}
                                \geq \left( \frac{1}{W_{A_{x_2}}} - \frac{1}{W_{A_x}}\right) D_{A_x}.
                        \end{align*}

                        This immediately gives us an upper bound on the minimum possible delay cost of $ A_x $
                        \begin{align}
                                \label{eq82}
                                D_{A_x}
                                \leq \frac{W_{A_{x_2}} (3 W_{A_{x_2}} - W_{A_{x_1}})}{2} C_{A_x}.
                        \end{align}

                        In this case, we split the arborescence $ A_x $ along the edge $ e' $ and connect the arborescences $ A_{x_1}, A_{x_2} $ and $ A_y $ separately to the root $ r $ as in Lemma \ref{lem99}. This leads to a total cost of at most
                        \begin{align*}
                                &\hspace{-1em} \underbrace{\left( 1 + \frac{W_{A_{x_1}}}{2}\right) C_{A_{x_1}} + \left( 1 + \frac{1}{W_{A_{x_1}}}\right) D_{A_{x_1}}}_{\text{cost bound for } A_{x_1}} + \underbrace{\left( 1 + \frac{W_{A_{x_2}}}{2}\right) C_{A_{x_2}} + \left( 1 + \frac{1}{W_{A_{x_2}}}\right) D_{A_{x_2}}}_{\text{cost bound for } A_{x_2}} \\
                                &\qquad + \underbrace{\left( 1 + \frac{W_{A_y}}{2}\right) C_{A_y} + \left( 1 + \frac{1}{W_{A_y}}\right) D_{A_y}}_{\text{cost bound for } A_y} \\
                          &\leq \left( 1 + \frac{W_{A_{x_1}}}{2}\right) C_{A_x} + \left( 1 + \frac{1}{W_{A_{x_2}}}\right) D_{A_x} + \left( 1 + \frac{W_{A_y}}{2}\right) C_{A_y} + \left( 1 + \frac{1}{W_{A_y}}\right) D_{A_y}\\
                          & =: \Phi.
                        \end{align*}

                        From $ W_{A_{x_1}} < 3 W_{A_{x_2}} $, we  get
                        \begin{align}
                                \label{eq81}
                                W_{A_{x_1}}
                                = \frac{3}{4} W_{A_{x_1}} + \frac{1}{4} W_{A_{x_1}}
                                < \frac{3}{4} W_{A_x}.
                        \end{align}

                        If we additionally assume that $ W_{A_{x_2}}\geq W_{A_y} $, we get
                        \begin{align*}
\Phi                                \overset{(\ref{eq81})}&{\leq} \left( 1 + \frac{3}{8} W_{A_x}\right) C_{A_x} + \left( 1 + \frac{1}{W_{A_y}}\right) D_{A_x} + \left( 1 + \frac{W_{A_y}}{2}\right) C_{A_y} + \left( 1 + \frac{1}{W_{A_y}}\right) D_{A_y} \\
                                &\leq \left( 1 + \max \left\{ \frac{3}{8} W_{A_x}, \frac{W_{A_y}}{2}\right\} \right) C_{A'} + \left( 1 + \frac{1}{\mu}\right) D_{A'} + \underbrace{\frac{\mu - W_{A_y}}{\mu W_{A_y}}}_{\geq 0} D_{A'} \\
                                \overset{(\ref{eq85})}&{\leq} \Bigl( 1 + \underbrace{\max \left\{ \frac{3}{8} W_{A_x}, \frac{W_{A_y}}{2}\right\} + \frac{\mu - W_{A_y}}{W_{A_y}}\cdot \frac{2\omega _{A'} - b\mu W_{A'}}{W_{A'} - \mu}}_{ = \max \{ \phi _2 (W_{A_x}, W_{A_y}), \phi _1 (W_{A_x}, W_{A_y})\}} \Bigr) C_{A'} + \left( 1 + \frac{1}{\mu}\right) D_{A'} \\
                                &\leq (1 + b\mu ) C_{A'} + \left( 1 + \frac{1}{\mu}\right) D_{A'}.
                        \end{align*}

                        Otherwise, if $ W_{A_{x_2}} < W_{A_y} $, we set $ \alpha ' := \left( \frac{1}{W_{A_{x_2}}} - \frac{1}{W_{A_y}}\right) \frac{W_{A_{x_2}} (3 W_{A_{x_2}} - W_{A_{x_1}})}{2} \geq 0 $. This gives us
                        \begin{align}
                                \label{eq80}
                                \frac{W_{A_{x_1}}}{2} + \alpha '
                                = \frac{3}{2} W_{A_{x_2}} - \frac{W_{A_{x_2}} (3 W_{A_{x_2}} - W_{A_{x_1}} )}{2 W_{A_y}} .
                        \end{align}

                        and
                        \begin{align*}
                            \Phi    \overset{(\ref{eq82})}&{\leq} \left( 1 + \frac{W_{A_{x_1}}}{2} + \alpha '\right) C_{A_x} + \left( 1 + \frac{W_{A_y}}{2}\right) C_{A_y} + \left( 1 + \frac{1}{W_{A_y}}\right) D_{A'} \\
                                &\leq \left( 1 + \max \left\{ \frac{W_{A_{x_1}}}{2} + \alpha ', \frac{W_{A_y}}{2}\right\} \right) C_{A'} + \left( 1 + \frac{1}{\mu}\right) D_{A'} + \underbrace{\frac{\mu - W_{A_y}}{\mu W_{A_y}}}_{\geq 0} D_{A'} \\
                                \overset{(\ref{eq85})}&{\leq} \!\Bigl( 1 + \underbrace{\max \left\{ \frac{W_{A_{x_1}}}{2} + \alpha ', \frac{W_{A_y}}{2}\right\} + \frac{\mu - W_{A_y}}{W_{A_y}}\cdot \frac{2\omega _{A'} - b\mu W_{A'}}{W_{A'} - \mu}}_{ \overset{(\ref{eq80})}{=} \max \{ \phi _3 (W_{A_x}, W_{A_{x_2}}, W_{A_y}), \phi _1 (W_{A_x}, W_{A_y})\}} \Bigr)C_{A'}\! +\! \left( 1 + \frac{1}{\mu}\right)\! D_{A'} \\
                                &\leq (1 + b\mu ) C_{A'} + \left( 1 + \frac{1}{\mu}\right) D_{A'}.
                        \end{align*}
                \end{subcase}
        \end{case}

        This concludes the proof, since in each case we satisfy the claimed cost bound.

        The running time analysis is almost identical to Lemma \ref{lem98}: We just compute
        the quantities delay weight, connection cost, minimum possible delay cost, and balance for up to two more sub-arborescences of $ A' $ .
\end{proof}

The question is, how small we can choose the parameter $ b$. We need $ b < \frac{2}{3} $  to obtain an improvement over
Lemma \ref{lem98}.

\begin{lem}
        The functions $ \phi _1 $, $ \phi _2 $ and $ \phi _3 $ satisfy the following properties
        \begin{enumerate}
                \item $ \phi _1 (x,y)\leq b\mu $ for all $ \mu \in \bR _{> 0} $ and $ (x,y)\in X^{\mu} $ $\quad\Longleftrightarrow\quad$ $ b\geq \frac{16}{27} $.

                \item  $ \phi _2 (x,y)\leq b\mu $ for all $ \mu \in \bR _{> 0} $ and $ (x,y)\in X^{\mu}$ $\quad\Longleftrightarrow\quad$ $ b\geq \frac{8}{13} $.

                \item  $ \phi _3 (x,x_2,y)\leq b\mu $ for all $ \mu \in \bR _{> 0}$ and $(x,y,x_2)\in X^{\mu}\!\times\! (0,\mu ] $ \\$\Longleftrightarrow$ $ b\geq \frac{1609\sqrt{1609} - 42427}{34992} $.
        \end{enumerate}
\end{lem}

\begin{proof}
  For each of the functions, we first   proof the statement under the premise  $b\ge \frac{1}{2} $ and
  then give a counter example for  $b<\frac{1}{2}$.

  If $b \ge \frac{1}{2}$, $ \phi _0 $ is monotonically non-decreasing in the first parameter, as
        \begin{align*}
                \frac{\partial \phi _0}{\partial x} (x,y)
                &= \frac{\mu - y}{y}\cdot \frac{(2y - b\mu )(x+y-\mu ) - (2xy - b\mu (x+y))\cdot 1}{(x+y-\mu )^2} \\
                &= \frac{\mu - y}{y}\cdot \frac{2y^2 - 2\mu y + b\mu ^2}{(x+y-\mu )^2} \\
                &= \frac{\mu - y}{y}\cdot \frac{\left( \sqrt{2} y - \frac{1}{\sqrt{2}} \mu\right) ^2 + \left( b - \frac{1}{2}\right) \mu ^2}{(x+y-\mu )^2} \\
                &\geq 0.
        \end{align*}

        This property transfers directly to the functions $ \phi _1, \phi _2 $ and $ \phi _3 $. We now look at the individual functions:
        \begin{enumerate}
                \item For  $ \phi _1 $ we obtain by  using the monotonicity property under the assumption $b \ge \frac{1}{2}$
                        \begin{align*}
                                \phi _1 (x,y)
                                &\leq \phi _1 (\mu ,y) \\
                                &= \frac{y}{2} + \left( \frac{\mu}{y} - 1\right) \left( (2-b)\mu - \frac{b\mu ^2}{y}\right) \\
                                &= \frac{y}{2} + (b-2)\mu + \frac{2\mu ^2}{y} - \frac{b\mu ^3}{y^2} \\
                                &= \frac{1}{2y^2} \left( y^3 - 4\mu y^2 + 4\mu ^2 y - \frac{32}{27} \mu ^3\right) + \left( \frac{16}{27} - b\right) \frac{\mu ^3}{y^2} + b\mu \\
                                &= \frac{1}{2y^2} \left( y^2 - \frac{4}{3} \mu y + \frac{4}{9} \mu ^2\right) \left( y - \frac{8}{3} \mu \right) + \left( \frac{16}{27} - b\right) \frac{\mu ^3}{y^2} + b\mu \\
                                &= \frac{1}{2y^2} \underbrace{\left( y - \frac{2}{3} \mu \right) ^2}_{\geq 0} \underbrace{\left( y - \frac{8}{3} \mu \right)}_{\leq 0} + \left( \frac{16}{27} - b\right) \frac{\mu ^3}{y^2} + b\mu .
                        \end{align*}

                        Since we can choose $ y = \frac{2}{3} \mu $, the first statement becomes evident.

                        From the last line in the equation, we also see that for   $b<  \frac{1}{2}$, $\phi_1(\mu,\frac{2}{3}\mu) > b\mu$.

                      \item For  $ \phi _2 $ we obtain by using the monotonicity property under the assumption $b \ge \frac{1}{2}$
                        \begin{align*}
                                \phi _2 (x,y)
                                &\leq \phi _2 (\mu ,y) \\
                                &= \frac{3}{8} \mu + (b-2) \mu + \frac{2\mu ^2}{y} - \frac{b\mu ^3}{y^2} \\
                                &= -\frac{13}{8} \mu + \frac{2\mu ^2}{y} - \frac{8}{13} \frac{\mu ^3}{y^2} + \left( \frac{8}{13} - b\right) \frac{\mu ^3}{y^2} + b\mu \\
                                &= -\frac{13 \mu}{8 y^2} \left( y^2 - \frac{16}{13} \mu y + \frac{64}{169} \mu ^2 \right) + \left( \frac{8}{13} - b\right) \frac{\mu ^3}{y^2} + b\mu \\
                                &= -\frac{13 \mu}{8 y^2} \underbrace{\left( y - \frac{8}{13} \mu \right) ^2}_{\geq 0} + \left( \frac{8}{13} - b\right) \frac{\mu ^3}{y^2} + b\mu .
                        \end{align*}

                        Since we can choose $ y = \frac{8}{13} \mu $, the second statement becomes evident.
                        From the last line in the equation, we also see that for   $b<  \frac{1}{2}$, $\phi_2(\mu,\frac{8}{13}\mu) > b\mu$.

                \item For $\phi _3 $ we also use the monotonicity property $ \phi _3 (x,x_2,y)\leq \phi _3 (\mu ,x_2,y) =: \psi (x_2,y) $ for $b\ge \frac{1}{2}$. Thus,
                        \begin{align*}
                                \psi (x_2,y)
                                &= \frac{3}{2} x_2 - \frac{4x_2^2 - \mu x_2}{2y} + (b-2)\mu + \frac{2\mu ^2}{y} - \frac{b\mu ^3}{y^2} \\
                                &= - \frac{2}{y} x_2^2 + \left( \frac{3}{2} + \frac{\mu}{2y}\right) x_2 + (b-2)\mu + \frac{2\mu ^2}{y} - \frac{b\mu ^3}{y^2} .
                        \end{align*}

                        For any fixed $ y $, $ \psi (\cdot , y) $ is a quadratic function with negative leading coefficient.
                        The first order condition is
                        \begin{align*}
                                \frac{\partial \psi}{\partial x_2} (x_2,y)
                                = - \frac{4}{y} x_2 + \frac{3}{2} + \frac{\mu}{2y}
                                \overset{!}{=} 0 \quad \Longleftrightarrow \quad x_2
                                = \frac{3}{8} y + \frac{1}{8} \mu
                        \end{align*}

                        Therefore, $ x_2 = \frac{3}{8} y + \frac{1}{8} \mu $ is a global maximum. Furthermore, because $ y\in (0,\mu ] $ we also have $ x_2\in (0,\mu ] $. To simplify the following calculation a bit, we set the following auxiliary quantity
                        \begin{align*}
                                \alpha
                                := \left( \frac{1609 \sqrt{1609} - 42427}{34992} - b\right) \frac{\mu ^3}{y^2} + b\mu.
                        \end{align*}

                        Now, we get
                        \begin{align*}
                                \psi (x_2,y)
                                &\leq \psi \left( \frac{3}{8} y + \frac{1}{8}\mu , y \right) \\
                                &= - \frac{2}{y} \left( \frac{3}{8} y + \frac{1}{8} \mu \right) ^2 + \left( \frac{3}{2} + \frac{\mu}{2y}\right) \left( \frac{3}{8} y + \frac{1}{8} \mu \right) + (b-2) \mu + \frac{2\mu ^2}{y} - \frac{b\mu ^3}{y^2} \\
                                &= - \frac{1}{32y}\! (9 y^2 + 6\mu y + \mu ^2)\! + \!\left( \frac{9}{16} y + \frac{3}{8} \mu + \frac{\mu ^2}{16 y}\right)\! + \!(b-2)\mu + \frac{2\mu ^2}{y} - \frac{b\mu ^3}{y^2} \\
                                &= \frac{9}{32} y - \frac{29}{16} \mu + \frac{65 \mu ^2}{32 y} - \frac{1609\sqrt{1609} - 42427}{34992} \frac{\mu ^3}{y^2} + \alpha \\
                                &= \frac{9}{32 y^2} \left( y^3 - \frac{58}{9} \mu y^2 + \frac{65}{9} \mu ^2 y - \frac{3218\sqrt{1609} - 84854}{19683} \mu ^3\right) + \alpha \\
                                &= \frac{9}{32 y^2} \left( y^2 - \frac{116 - 2\sqrt{1609}}{27} \mu y + \frac{4973 - 116\sqrt{1609}}{729} \mu ^2\right) \\
                                &\qquad \cdot \left( y - \frac{58 + 2\sqrt{1609}}{27} \mu \right) + \alpha \\
                                &= \frac{9}{32 y^2} \underbrace{\left( y - \frac{58 - \sqrt{1609}}{27} \mu \right) ^2}_{\geq 0} \underbrace{\left( y - \frac{58 + 2\sqrt{1609}}{27} \mu \right)}_{\leq 0} + \alpha .
                        \end{align*}

                        Since we can choose $ y = \frac{58 - \sqrt{1609}}{27} \mu \approx 0.7\mu $, the 3rd statement becomes evident for $b\ge \frac{1}{2}$.

                        If $b< \frac{1}{2}$, $\phi_3(\mu, \frac{\mu}{4},\frac{\mu}{2})= (\frac{19}{8}-3b) \mu > (\frac{19}{8}- \frac{3}{2}) \mu >  \frac{1}{2}\mu > b\mu$.
        \end{enumerate}
\end{proof}

We now obtain immediately the following improvement of Lemma \ref{lem98}:

\begin{cor}
        \label{col99}
        Let $ T'\in \cT $ and let $ A' $ be the sub-arborescence stored with $ T' $. Then the terminals from $ T' $ can be connected to the root $ r $ in $ \mathcal{O} (|E(A')| + |T_{A'}|) $ time with total cost at most
        \begin{align*}
                \left( 1 + \frac{1609\sqrt{1609} - 42427}{34992}\mu \right) C_{A'} + \left( 1 + \frac{1}{\mu}\right) D_{A'}.
        \end{align*}
\end{cor}

\subsection{The main result}

\label{sec:main-result}

\begin{thm}
        \label{theo97}
        Given an instance $(M,c,r,T,p,w)$ of the \textsc{Uniform Cost-Distance Steiner Tree Problem}, we can compute in $ \mathcal{O} (\Lambda + |T|) $ time a Steiner tree with objective value at most
        \begin{align}
                \label{eq96}
                \left( 1 + \frac{1609\sqrt{1609} - 42427}{34992}\mu \right) C + \left( 1 + \frac{1}{\mu}\right) D,
        \end{align}

        where $ C $ is the cost of a $ \beta $-approximate minimum-length Steiner tree and $D = D_T$. Here is $ \Lambda $ the running time for computing a $ \beta $-approximate minimum Steiner tree for $ T\cup \{ r\} $.
\end{thm}
\begin{proof}
We run Algorithm~\ref{alg:khazraei21} with two modifications:
\begin{enumerate}
\item The arborescence $A_r$ containtaining the  root $ r $ after Step 2 may be  re-connected to the root $ r $ according to the proof of Lemma~\ref{lem97}.

\item The arborescences in $\cA$ that were cut off in Step 2  are re-connected to the root $ r $ as in Corollary~\ref{col99}.
\end{enumerate}
Observe that  the bound in Lemma~\ref{lem97} is dominated by the bound in Corollary~\ref{col99}, as $ \frac{1609\sqrt{1609} - 42427}{34992} > \frac{1}{2}$.
Finally, the total cost of the computed solution is  upper bounded  by the sum of the cost bounds for these $r$-arborescences, which is (\ref{eq96}).

For the running time analysis, we consider the individual steps of the  algorithm:

In Step 1, a $ \beta $-approximate minimum Steiner tree for $ T\cup \{ r\} $ is computed in time $ \mathcal{O} (\Lambda ) $ and transformed into the arborescence $ A_0 $ in linear time.

As in \cite{Khazraei21} we compute an initial branching in $ \mathcal{O}(|E(A_0)|) \leq \mathcal{O} (|T|) $ in Step 2.
Using Corollary~\ref{col99} and Lemma~\ref{lem97} we improve the initial branching and connect the resulting sub-arborescences to the root with a total running time of $ \mathcal{O}(|E(A_0)| + |T|)\leq \mathcal{O}(|T|) $.
\end{proof}

Finally, we choose the threshold $ \mu $ based on the  quantities $ C $ and $ D $.
Here, we consider a more general formulation that links potential  further improvement of Theorem \ref{theo97} to  a further improvement of the approximation ratio.
\begin{thm}
        \label{theo96}
        Let $ b\in \bR $ be a constant and $ \mathbf{A} $ be a polynomial algorithm for the following input
        \begin{itemize}
                \item an instance  $(M,c,r,T,p,w)$  of the \textsc{Uniform Cost-Distance Steiner Tree Problem},
                \item a $ \beta $-approximate minimum Steiner tree for $T\cup \{r\}$ with cost $ C $, and

                \item a threshold $ \mu > 0 $.
        \end{itemize}

        Let $\mathbf{A}$   compute a solution to the \textsc{Uniform Cost-Distance Steinerbaum Problem} with total cost of at most
        $$
        (1 + b\mu ) C + \left( 1 + \frac{1}{\mu}\right) D,
        $$
        where $D = D_T = \sum_{t\in T} w(t)c(r,t)$.

        Then,  algorithm $ \mathbf{A} $ can be modified (in particular by the choice of $\mu$) to  solve the \textsc{Uniform Cost-Distance Steiner Tree Problem} in polynomial time with an approximation factor of
        $$\beta + \frac{2b \beta}{\sqrt{4b\beta + (\beta - 1)^2} + \beta - 1}.$$
\end{thm}

\begin{proof}
        We make the following changes to the algorithm $ \mathbf{A} $:
      If  $ C = 0 $, each $ r $-$ t $-path, $ t \in T $,  has length $ 0 $ in the initial Steiner tree. So this is already an optimal solution and we do not run the algorithm any further.

      Else if   $ D  = 0 $, we have $ w(t) = 0 $ or $ c(r,t) = 0 $ for each terminal $ t\in T $. For each terminal $ t\in T $ with $c (r,t) = 0 $ we add a shortest $ r $-$ t $-path to the initial computed Steiner arborescence. We then determine a shortest-path arborescence rooted at $ r $ in the resulting graph.
      We obtain a new $ \beta $-approximate minimal Steiner tree with delay cost $ 0 $. In particular, we have found a $ \beta $-approximate minimum cost-distance Steiner tree.

      Otherwise, set $ \mu := \sqrt{\frac{D}{bC}} $ and run algorithm $ \mathbf{A} $. Then we obtain a solution with total cost at most
        \begin{align*}
                C + D + 2\sqrt{bCD}
                \leq \beta C_{SMT} + D + 2\sqrt{b\beta C_{SMT} D},
        \end{align*}
        where $C_{SMT}$ is the length of a minimum Steiner tree for $T\cup \{r\}$.
        $ C_{SMT} + D $ is a trivial lower bound on the total cost of any solution. Therefore, the approximation factor is at most the maximum of the function $ f\colon \bR _{>0}^2 \to \bR $ given by
        \begin{align*}
                f(x,y)
                := \frac{\beta x + y + 2\sqrt{b\beta x y}}{x + y}
                = \beta + \frac{(1-\beta )y + 2\sqrt{b\beta x y}}{x + y}.
        \end{align*}

        To determine the maximum of $ f $, we first  set
                $$a
                := \frac{2b\beta}{\sqrt{4b\beta + (\beta - 1)^2} + \beta - 1}$$
        and get
        \begin{align*}
                \frac{b\beta}{a} + 1 - \beta - a
                &= \frac{1}{2} \sqrt{4b\beta + (\beta - 1)^2} - \frac{1}{2} (\beta - 1) - \frac{2b\beta}{\sqrt{4b\beta + (\beta - 1)^2} + \beta - 1} \\
                &= \frac{(\sqrt{4b\beta + (\beta - 1)^2} - (\beta - 1)) (\sqrt{4b\beta + (\beta - 1)^2} + (\beta - 1)) - 4b\beta}{2 (\sqrt{4b\beta + (\beta - 1)^2} + \beta - 1)} \\
                &= 0.
        \end{align*}

        Therefore,
        \begin{align*}
                f(x,y)
                &= \beta + \frac{(1-\beta ) y + 2\sqrt{b\beta xy}}{x + y} - \frac{\left( \frac{b\beta}{a} + 1-\beta - a\right) y}{x + y} \\
                &= \beta + \frac{2\sqrt{b\beta xy}}{x + y} - \frac{\left( \frac{b\beta}{a} - a\right) y}{x + y} \\
                &= \beta + a - \frac{ax}{x+y} + \frac{2\sqrt{b\beta xy}}{x + y} - \frac{\frac{b\beta}{a} y}{x + y} \\
                &= \beta + a - \frac{a}{x+y} \left( \sqrt{x} - \frac{\sqrt{b\beta}}{a} \sqrt{y}\right) ^2.
        \end{align*}

        As $ a\geq 0 $, we obtain
        \begin{align*}
                f(x,y)
                \leq \beta + a
                = \beta + \frac{2b\beta}{\sqrt{4b\beta + (\beta - 1)^2} + \beta - 1},
        \end{align*}
        proving the claimed approximation ratio.
\end{proof}

Applying Theorem~\ref{theo96} with $ b = \frac{2}{3}$, which is implied by   Lemma~\ref{lem98} and Lemma~\ref{lem97}, and the best known choice of $  b = \frac{1609\sqrt{1609} - 42427}{34992} $ in Theorem \ref{theo97}, we obtain the  approximation factors shown in Table~\ref{table:approx-factors} (rounded to five decimal digits) for some interesting values of $ \beta $:

\begin{table}[t]
        \centering
        \caption{Comparison of approximation factors for the \textsc{Uniform Cost-Distance Steiner Tree Problem} with different apprimation factors $ \beta $ for the minimum-length Steiner tree problem.}
        \begin{tabular}{l cccc}
                \toprule
                Parameter $ \beta $ & $ 1 $ & $ \ln (4) + \epsilon $ & $ \frac{3}{2} $ & $ 2 $ \\
                \midrule
                Algorithm~\ref{alg:khazraei21} \cite{Khazraei21}& $ 2.00000 $ & $ 2.38630 $ & $ 2.50000 $ & $ 3.00000 $ \\
                 Theorem~\ref{theo96} with $b=\frac{2}{3}$ & $ 1.81650 $ & $ 2.17371 $ & $ 2.28078 $ & $ 2.75831 $ \\
                 Theorem~\ref{theo96} with best $b$   & $ 1.79497 $ & $ 2.14887 $ & $ 2.25522 $ & $ 2.73042 $ \\
                \bottomrule
        \end{tabular}
        \label{table:approx-factors}
\end{table}

\section{Conclusion}
\label{sec:conclusion}

By inspecting worst-case instances for the approximation algorithm of
\cite{Khazraei21}, we could improve that algorithm and  the
approximation factor for the \textsc{Uniform Cost-Distance Steiner Tree
Problem}.

This was achieved by sub-dividing arborescences that were cut off even
further into up to 3 components.  A further improvement can be
expected by sub-dividing into more components. But the analysis would
probably become quite exhaustive.

Our algorithm is extremely fast. After computing an approximate
minimum-length Steiner tree, the remaining steps run in linear time,
which previously took a quadratic running time.

\printbibliography

\appendix

\section{A tight example for Theorem~\ref{theo99} in the Manhattan plane}
\label{appendix:tight-example-in-Manhattan-plane}

\begin{proof}[Alternative proof of Theorem~\ref{lem95}]
  We create an instance class in the Manhattan plane $(M,c) = (\mathbb{R}^2,\ell_1)$.

    Let $ k\in \bN _{>0} $. We set $ \epsilon ' := \frac{1}{k} $ and consider the following instance in $(\mathbb{R}^2,\ell_1)$:
    \begin{itemize}
                \item The root $ r $ lies at the point $ p(r) = (0,0)\in \bR ^2 $.

                \item $T$ contains  $ 4k + $ 7  terminals. $4k-1$ terminals  $v_i$ ($i=\{1,\dots, 4k-1\}$ placed at
                        \begin{align*}
                                p(v_i)
                                = \begin{cases}
                                        (i,0) & \text{if } i\in \{ 1,\ldots, k\} , \\
                                        (k,i-k) & \text{if } i\in \{ k+1,\ldots, 2k\} , \\
                                        (3k-i,k) & \text{if } i\in \{ 2k+1,\ldots, 3k\} , \\
                                        (0,4k-i) & \text{if } i\in \{ 3k+1,\ldots, 4k-1\} ,
                                \end{cases}
                        \end{align*}

                        and 8 terminals $w_1,\dots, w_k$, placed at
                        \begin{align*}
                                p(w_i)
                                = \begin{cases}
                                        (k-1,-i\epsilon ') & \text{if } i\in \{ 1,2\} , \\
                                        (k, - i\epsilon ') & \text{if } i\in \{ 3,4\} , \\
                                        (- (i-4)\epsilon ', k-1) & \text{if } i\in \{ 5,6\} , \\
                                        (- (i-4)\epsilon ', k) & \text{if } i\in \{ 7,8\} .
                                \end{cases}
                        \end{align*}

                \item The terminals have the following delay weights:
                        \begin{align*}
                                w(t)
                                = \begin{cases}
                                        1 & \text{if } t = v_1 \text{ or } t = v_{4k-1}, \\
                                        0 & \text{if } t\in \{ v_2, v_3, \ldots , v_{4k-2}\} , \\
                                        \frac{1}{2} + \frac{1}{k} & \text{if } t\in \{ w_1, w_2, \ldots , w_k\} .
                                \end{cases}
                        \end{align*}

                \item The distances are given by the Manhattan metric $\ell_1$.
        \end{itemize}

        \begin{figure}
                \centering
                \begin{tikzpicture}
                        \node[fsquare,label=225:$ r $] (r) at (0,0) {};
                        \node[fcircle,label=above:$ v_1 $] (v0) at (1,0) {};
                        \node[] (v1) at (2,0) {$ \ldots $};
                        \node[fcircle,label=above:$ v_ {k-1} $] (v2) at (3,0) {};
                        \node[fcircle,label=right:$ v_k $] (v3) at (4,0) {};
                        \node[fcircle,label=right:$ v_{k+1} $] (v4) at (4,1) {};
                        \node (v5) at (4,2) {$ \vdots $};
                        \node[fcircle,label=right:$ v_{2k-1} $] (v6) at (4,3) {};
                        \node[fcircle,label=right:$ v_{2k} $] (v7) at (4,4) {};
                        \node[fcircle,label=above:$ v_{2k+1} $] (v8) at (3,4) {};
                        \node (v9) at (2,4) {$ \ldots $};
                        \node[fcircle,label=above:$ v_{3k-1} $] (v10) at (1,4) {};
                        \node[fcircle,label=above:$ v_{3k} $] (v11) at (0,4) {};
                        \node[fcircle,label=right:$ v_{3k+1} $] (v12) at (0,3) {};
                        \node (v13) at (0,2) {$ \vdots $};
                        \node[fcircle,label=right:$ v_{4k-1} $] (v14) at (0,1) {};

                        \node[fcircle,label=right:$ w_1 $] (w1) at (3,-0.5) {};
                        \node[fcircle,label=right:$ w_2 $] (w2) at (3,-1) {};
                        \node[fcircle,label=right:$ w_3 $] (w3) at (4,-1.5) {};
                        \node[fcircle,label=right:$ w_4 $] (w4) at (4,-2) {};

                        \node[fcircle,label=above:$ w_5 $] (w5) at (-0.5,3) {};
                        \node[fcircle,label=above:$ w_6 $] (w6) at (-1,3) {};
                        \node[fcircle,label=above:$ w_7 $] (w7) at (-1.5,4) {};
                        \node[fcircle,label=above:$ w_8 $] (w8) at (-2,4) {};

                        \draw (v7) to (v8) to (v9) to (v10) to (v11) to (v12) to (v13) to (v14) to (r) to (v0) to (v1) to (v2) to (v3) to (v4) to (v5) to (v6);

                        \draw (v2) to (w1) to (w2);
                        \draw (v3) to (w3) to (w4);

                        \draw (v11) to (w7) to (w8);
                        \draw (v12) to (w5) to (w6);

                \end{tikzpicture}
                \caption{An optimal cost-distance Steiner tree}
                \label{fig89}
        \end{figure}
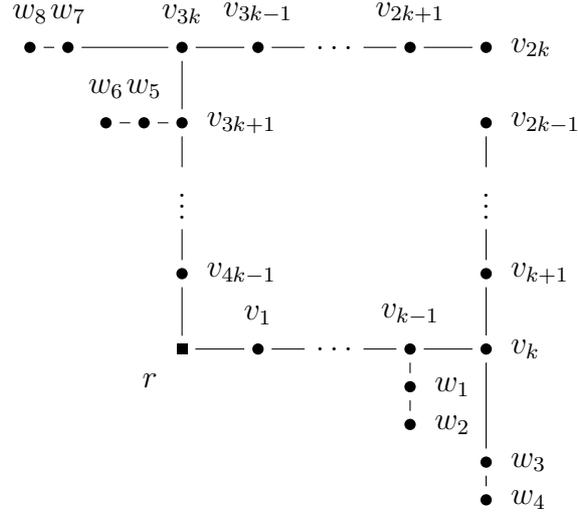

        Figure \ref{fig89} shows  an optimal Steiner tree for this instance if $ k\geq 2 $.
        It also minimizes the path length from the root $ r $ to each terminal, so it is an optimal cost-distance Steiner tree. Using
        \begin{align*}
                C_{SMT}
                &= 4k-1 + 12\epsilon '
                = 4k-1 + \frac{12}{k}
        \end{align*}

        and
        \begin{align*}
                D_{SP}
                &= 2 + 2\left( \frac{1}{2} + \frac{1}{k}\right) (2(k-1) + 3\epsilon ' + 2k + 7\epsilon ') \\
                &= 2 + \left( 1 + \frac{2}{k}\right) \left( 4k - 2 + \frac{10}{k}\right) \\
                &= 4k + 8 + \frac{6}{k} + \frac{20}{k^2}
        \end{align*}

        its total costs is therefore
        \begin{align}
                \label{eq86}
                C_{SMT} + D_{SP}
                = 8k + 7 + \frac{18}{k} + \frac{20}{k^2} .
        \end{align}

        So Algorithm~\ref{alg:khazraei21} with the minor improvements  from Section~\ref{sec:worst_case_example} chooses the threshold $ \mu $ as
        \begin{align*}
                \mu
                = \sqrt{\frac{D_{SP}}{C_{SMT}}}
                = \sqrt{\frac{4k + 8 + \frac{6}{k} + \frac{20}{k^2}}{4k - 1 + \frac{12}{k}}}.
        \end{align*}

        In particular, for $ k\geq 2 $ we  see that $ \mu \geq 1\geq \frac{1}{2} + \frac{1}{k} $ holds.
        As an upper bound we get
        \begin{align*}
                \mu
                &= \sqrt{1 + \frac{9 - \frac{6}{k} + \frac{20}{k^2}}{4k - 1 + \frac{12}{k}}} \\
                &< \sqrt{1 + \frac{16 - \frac{4}{k} + \frac{48}{k^2}}{4k - 1 + \frac{12}{k}}} \\
                &= \sqrt{1 + \frac{4}{k}} \\
                &< \sqrt{\left( 1 + \frac{2}{k}\right) ^2} \\
                &= 2\left( \frac{1}{2} + \frac{1}{k}\right).
        \end{align*}

        So if we remove the edge $ \{ r, v_1\} $ in Figure \ref{fig89} and add the edge $ \{ v_{2k-1}, v_{2k}\} $ instead, it is still an optimal Steiner tree.
        If this is computed in Step 1 of the Algorithm~\ref{alg:khazraei21}, the  cost-distance Steiner tree in Figure~\ref{fig:worst-case-example-sol-manhattan} can be computed with the above $\mu$:

        \begin{figure}
                \centering
                \begin{tikzpicture}
                        \node[fsquare,label=225:$ r $] (r) at (0,0) {};
                        \node[fcircle,label=above:$ v_1 $] (v0) at (1,0) {};
                        \node[] (v1) at (2,0) {$ \ldots $};
                        \node[fcircle,label=above:$ v_ {k-1} $] (v2) at (3,0) {};
                        \node[fcircle,label=right:$ v_k $] (v3) at (4,0) {};
                        \node[fcircle,label=right:$ v_{k+1} $] (v4) at (4,1) {};
                        \node (v5) at (4,2) {$ \vdots $};
                        \node[fcircle,label=right:$ v_{2k-1} $] (v6) at (4,3) {};
                        \node[fcircle,label=right:$ v_{2k} $] (v7) at (4,4) {};
                        \node[fcircle,label=above:$ v_{2k+1} $] (v8) at (3,4) {};
                        \node (v9) at (2,4) {$ \ldots $};
                        \node[fcircle,label=above:$ v_{3k-1} $] (v10) at (1,4) {};
                        \node[fcircle,label=above:$ v_{3k} $] (v11) at (0,4) {};
                        \node[fcircle,label=right:$ v_{3k+1} $] (v12) at (0,3) {};
                        \node (v13) at (0,2) {$ \vdots $};
                        \node[fcircle,label=right:$ v_{4k-1} $] (v14) at (0,1) {};

                        \node[fcircle,label=right:$ w_1 $] (w1) at (3,-0.5) {};
                        \node[fcircle,label=right:$ w_2 $] (w2) at (3,-1) {};
                        \node[fcircle,label=right:$ w_3 $] (w3) at (4,-1.5) {};
                        \node[fcircle,label=right:$ w_4 $] (w4) at (4,-2) {};

                        \node[fcircle,label=above:$ w_5 $] (w5) at (-0.5,3) {};
                        \node[fcircle,label=above:$ w_6 $] (w6) at (-1,3) {};
                        \node[fcircle,label=above:$ w_7 $] (w7) at (-1.5,4) {};
                        \node[fcircle,label=above:$ w_8 $] (w8) at (-2,4) {};

                        \draw (v0) to (v1) to (v2) to (v3) to (v4) to (v5) to (v6) to (v7) to (v8) to (v9) to (v10) to (v11) to (v12) to (v13) to (v14) to (r);

                        \draw (r)++(0.03,-5pt) to (0.03,-0.5) to (w1) to (w2);
                        \draw (r)++(-0.03,-5pt) to (-0.03,-1.5) to (w3) to (w4);

                        \draw (r)++(-5pt,0.03) to (-0.5,0.03) to (w5) to (w6);
                        \draw (r)++(-5pt,-0.03) to (-1.5,-0.03) to (w7) to (w8);
                \end{tikzpicture}
                \caption{A computed cost-distance Steiner tree}
                \label{fig:worst-case-example-sol-manhattan}
        \end{figure}
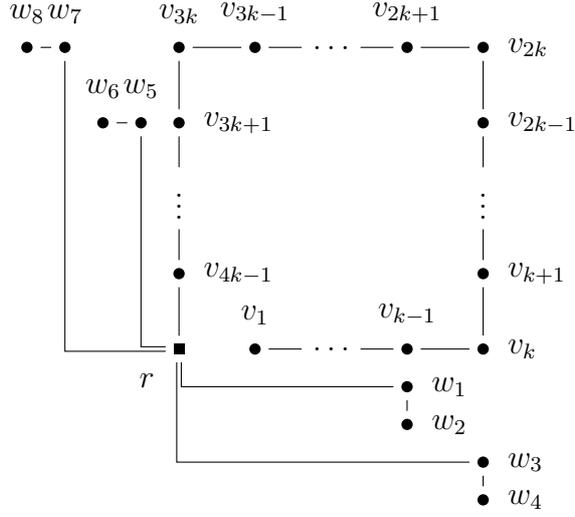

        This solution has a total cost of
        \begin{align*}
                &\hspace{-1em}(4k-1 + 4k-2 + 12\epsilon ') + \Bigl( 4k + \underbrace{2\left( \frac{1}{2} + \frac{1}{k}\right) (2(k-1) + 3\epsilon ' + 2k + 7\epsilon '}_{= 4k + 6 + \frac{6}{k} + \frac{20}{k^2}}\Bigr) \\
                &= 16k + 3 + \frac{18}{k} + \frac{20}{k^2} .
        \end{align*}

        It exceeds  the value (\ref{eq86}) of the optimum solution asymptotically by a factor of $2$ for increasing $k$.
\end{proof}

\end{document}